\documentclass[11pt,a4paper]{article}
\usepackage[pagebackref]{hyperref}
\usepackage[a4paper,margin=2.5cm]{geometry}
\usepackage{amssymb,amsmath,amsthm}
\usepackage{framed}

\DeclareMathOperator{\p}{\mathbb{P}}
\DeclareMathOperator{\pp}{\mathbf{P}}
\DeclareMathOperator{\States}{{\mathcal{S}}}

\DeclareMathOperator{\Cov}{{\rm Cov}}

\newcommand{\Var}{\mathrm{Var}\,}

\newcommand{\E}{\mathbb{E}}

\newcommand{\C}{\mathbb{C}}

\newcommand{\R}{\mathbb{R}}

\newcommand{\id}{\mathrm{Id}}
\newcommand{\tr}{\mathrm{Tr}}
\newcommand{\la}{\langle}
\newcommand{\ra}{\rangle}

\newcommand{\unif}{\mathrm{Unif}}
\newcommand{\Ha}{\mathcal{H}}

\newtheorem{lemma}{Lemma}[section]
\newtheorem{theorem}[lemma]{Theorem}
\newtheorem{prop}[lemma]{Proposition}
\newtheorem{remark}[lemma]{Remark}
\newtheorem{corollary}[lemma]{Corollary}

\newtheorem{example}[lemma]{Example}
\newtheorem{defi}[lemma]{Definition}

\newcommand{\Ca}{C_1}
\newcommand{\Cb}{C_2}
\newcommand{\Cc}{C_3}
\newcommand{\Cd}{C_4}
\newcommand{\Ce}{C_6}
\newcommand{\Cf}{C_7}
\newcommand{\Cg}{C_8}
\newcommand{\Ch}{C_9}
\newcommand{\Ci}{C_{12}}
\newcommand{\Cj}{C_{10}}
\newcommand{\Ck}{C_{13}}
\newcommand{\Cl}{C_{11}}
\newcommand{\Cm}{C_{15}}
\newcommand{\Cn}{C_{16}}
\newcommand{\Co}{C_5}
\newcommand{\Cp}{C_{14}}
\title{Metric and classical fidelity uncertainty relations for random unitary matrices}

\author{Rados{\l}aw Adamczak\thanks{Institute of Mathematics, University of Warsaw, Banacha 2, 02-097 Warszawa, Poland. \mbox{E-mail: R.Adamczak@mimuw.edu.pl}. Research partially supported by the NCN Grant 2015/18/E/ST1/00214.}}

\begin{document}
\maketitle

\begin{abstract}
We analyze uncertainty relations on finite dimensional Hilbert spaces expressed in terms of classical fidelity, which are stronger then metric uncertainty relations introduced by Fawzi, Hayden and Sen. We establish validity of fidelity uncertainty relations for random unitary matrices with optimal parameters (up to universal constants) which improves upon known results for the weaker notion of metric uncertainty.

This result is then applied to locking classical information in quantum states and allows to obtain optimal locking in Hellinger distance, improving upon previous results on locking in the total variation distance, both by strengthening the metric used and by improving the dependence on parameters.

We also show that general probabilistic estimates behind the main theorem can be used to prove existence of data hiding schemes with Bayesian type guarantees.

As a byproduct of our approach we obtain existence of almost Euclidean subspaces of the matrix spaces $\ell_1^n(\ell_2^m)$ with a better dimension/distortion dependence than allowed in previously known constructions.

\end{abstract}

\section{Introduction}

\subsection{General overview} Uncertainty relations, being of fundamental importance in quantum theory, continue to play central role in quantum information processing. One particular instance where uncertainty relations have proved extremely useful is information locking \cite{PhysRevLett.92.067902,MR2932028}, a phenomenon without a counterpart in classical information theory, allowing for communication protocols in which sending just a few bits of classical information between two parties provides one of them with a potentially unbounded increase of information about some shared resource.  This allows e.g. for building quantum encryption systems with strong information-theoretical security guarantees in which the key is much shorter than the message (which is in contrast with Shannon's theorem of classical information theory, asserting that in classical communication the length of the key must be at least as long as the message itself). We postpone a more precise description of this phenomenon to Section \ref{sec:locking} and here we just mention that protocols of this type have been recently successfully implemented in experiments \cite{2016arXiv160506556L}.

The focus on discrete information requires new approaches to quantify uncertainty, different from the original variance approach used by Heisenberg, Kennard and Robertson. Among particularly useful ways of expressing the uncertainty for sets of $t$ measurements performed on finite-state quantum systems is the average Shannon entropy per measurements
\begin{displaymath}
\frac{1}{t}\sum_{k=1}^{t} H(p_{k,|\psi\ra}),
\end{displaymath}
where $p_{k,|\psi\ra}$ is the probability measure describing the outcome of the $k$-th one among $d$ measurements on the system in state $|\psi\ra$ and $H$ is the Shannon entropy, defined for probability vectors $p = (p(1),\ldots,p(d))$ as
\begin{displaymath}
H(p) = - \sum_{i=1}^{d}  p(i)\log p(i).
\end{displaymath}

State independent lower bounds on the average entropy of the form
\begin{align}\label{eq:eur}
\inf_{|\psi\ra} \frac{1}{t}\sum_{k=1}^{t} H(p_{k,|\psi\ra}) \ge c(t,d),
\end{align}
where $d$ is the dimension of the underlying Hilbert space, are called entropic uncertainty principles and play an important role in quantum information. The infimum above is taken over the set of all pure states $|\psi\ra$ of the system. According to the well known inequality by Maassen and Uffink \cite{MR932170} for two projective measurements in mutually unbiased bases they hold with the optimal constant $c(2,d) = \frac{1}{2}\log d$. For more than two measurements explicit constructions are not known and one instead relies on probabilistic approach introduced for the first time by Hayden et al. in \cite{MR2094521}. Recently Adamczak et al. \cite{ALPZ}, answering a question by Leung, Wehner and Winter \cite{MR2602484}, have proved that i.i.d. random unitary matrices in high dimensions with high probability satisfy \eqref{eq:eur} with $c(t,d) = (1-1/t)\log d - C$ for a numerical constant $C$. This is asymptotically, for $d\to \infty$, equivalent to the best possible bound $(1-1/t)\log d$. We refer to the survey article \cite{MR2602484} and more recent \cite{2015arXiv151104857C} for a detailed discussion of entropic uncertainty relations and their applications.

In order to pass to uncertainty relations which we are going to investigate in this article, let us note that, as one can easily see, the inequality \eqref{eq:eur} with $c(t,d) = c(t)\log d$ can be equivalently interpreted as
\begin{align}\label{la:eur-kl}
\sup_{|\psi\ra} \frac{1}{t}\sum_{k=1}^{t} D_{KL}(p_{k,|\psi\ra},\unif([d])) \le (1-c(t))\log d,
\end{align}
where $\unif([d])$ is the uniform measure on the set $[d] = \{1,\ldots,d\}$ of possible outcomes of the measurement and $D_{KL}$ stands for the Kullback-Leibler divergence, i.e. for two probability vectors $p = (p(1),\ldots,p(d))$, $q = (q(1),\ldots,q(d))$,
\begin{displaymath}
D_{KL}(p,q) = \sum_{i=1}^{d} p(i) \log \Big(\frac{p(i)}{q(i)}\Big).
\end{displaymath}

Note that $\max_p D_{KL}(p,\unif([d])) = \log d$ (the maximum is achieved for $p$ being a Dirac's mass at some $i \in [d]$), therefore the bound $c(t) \simeq 1 - 1/t - O(1/\log d))$ obtained in \cite{ALPZ} shows that as the number of measurements and the dimension $d$ grow, the average `distance' between $p_{k,|\psi\ra}$ and $\unif([d])$ becomes very small when compared to the `diameter' of the space of probability measures with respect to the Kullback-Leibler divergence.

Such interpretation suggests the possibility that in measuring uncertainty one can use also other measures of proximity than the Kullback-Leibler divergence, which may correspond in a better way to the applications at hand. Such an approach was indeed proposed in \cite{MR2932028} by Fawzi, Hayden and Sen, who used the total-variation distance
\begin{align}
D_{TV}(p,q) = \frac{1}{2}\sum_{i=1}^d |p(i) - q(i)|.
\end{align}
Since the diameter of the probability simplex in the total variation distance is equal to one, the corresponding uncertainty principle should be of the form
\begin{align}\label{la:mur-tv}
\sup_{|\psi\ra} \frac{1}{t}\sum_{k=1}^{t} D_{TV}(p_{k,|\psi\ra},\unif([d])) \le f(t)
\end{align}
for some function $f(t)$ converging to zero with $t \to \infty$ (in particular independent of the dimension).

However, to obtain a meaningful uncertainty relation the Authors of \cite{MR2932028} had to depart from the setting of projective measurements and use an ancilla system. More precisely they introduce the following setting.

Let $\Ha_A$, $\Ha_B$ be two finite-dimensional Hilbert spaces of dimension $d_A$ and $d_B$ respectively. Let $|a\ra^A$, $a = 1,\ldots,d_A$ and $|b\ra^B$, $b=1,\ldots,d_B$ be fixed orthonormal bases of $\Ha_A$ and $\Ha_B$. Throughout the article we will always use the notation $\Ha= \Ha_A\otimes \Ha_B$.

For $a \in [d_A]$ and $|\psi\ra \in \Ha$ define $p^A_{|\psi\ra} = (p^A_{|\psi\ra}(a))_{a\in [d_A]}$ by
\begin{align}\label{eq:p^A}
p^A_{|\psi\ra}(a) = \sum_{b=1}^{d_B} |\la a|^A\la b|^B|\psi\ra|^2 = \tr \Big(|a\ra^A\la a| (\tr_B|\psi\ra\la\psi|)\Big),
\end{align}
where $\tr_B = \id_{\Ha_A}\otimes \tr$ is the partial trace with respect to the system $\Ha_B$.

Thus $p^A_{|\psi\ra}$ describes the outcome distribution for the local projective measurement in the basis $\{|a\ra^A\}$ performed on the $\Ha_A$ part of the bipartite system $\Ha$ in state $|\psi\ra$.

In what follows we will use the notation $d = d_Ad_B$ for the dimension of the space $\Ha = \Ha_A \otimes \Ha_B$.

We are now in position to state the definition of metric uncertainty relation introduced in \cite{MR2932028}.

\begin{defi}[Metric uncertainty relations] Let $U_1,\ldots,U_t$ be unitary transformations of $\Ha = \Ha_A\otimes \Ha_B$. For $\varepsilon > 0$ we will say that $\{U_k\}$ satisfy the $\epsilon$-metric uncertainty relation on $\Ha_A$ if
\begin{align}\label{eq:mur}
\sup_{|\psi\ra \in S_\Ha}\frac{1}{t} \sum_{k=1}^t D_{TV}(p_{U_k|\psi\ra}^A,\unif{[d_A]}) \le \varepsilon,
\end{align}
where $S_\Ha$ is the unit sphere of $\Ha$.
\end{defi}

\begin{remark}
The supremum above is taken over the set of pure states on $\Ha$, but by convexity it is easy to see that one can equivalently take the supremum over the set of all states.
\end{remark}

\begin{remark}
In the setting of Fawzi-Hayden-Sen one considers a fixed measurement performed on $t$ different states $U_k|\psi\ra$ of the system, which may seem different from the setting of entropic uncertainty relations we started with. Note however, that if $U$ is a unitary transformation, then the distribution describing a projective measurement on a system in state $U|\psi\ra$ performed in the computational basis $\{|i\ra\}$ coincides with the distribution describing the measurement of a system in state $|\psi\ra$ performed in the basis $\{U^\dag |i\ra\}$, so entropic uncertainty relations can also be written as uncertainty relations involving a fixed measurement and several unitary transformations of the state.
\end{remark}

In \cite{MR2932028} Fawzi, Hayden and Sen proved the following theorem concerning the existence of matrices satisfying metric uncertainty relations.

\begin{theorem}[Fawzi-Hayden-Sen, \cite{MR2932028}]\label{th:Fawzi-Hayden-Sen}
Let $U_1,\ldots,U_t$ be i.i.d. random unitary transformations distributed according to the Haar measure on the unitary group of $\Ha$. If $\varepsilon \in (0,1)$, $d_B \ge 9/\varepsilon^2$ and $t > \frac{72\cdot 16\ln(9/\varepsilon)}{\varepsilon^2}$, then with probability at least
\begin{displaymath}
1 - 4\exp\Big(-d\Big(\frac{\varepsilon^2t}{144}-2\ln(9/\varepsilon)\Big)\Big),
\end{displaymath}
$U_1,\ldots,U_{t}$ satisfy the $\varepsilon$ metric uncertainty relation on $\Ha_A$.
\end{theorem}

The Authors of \cite{MR2932028} provide also more explicit constructions of matrices satisfying metric uncertainty relations, however with a worse dependence of $d_B$ and $t$ on the parameter $\varepsilon$ and additional dependence on the dimension $d_A$. They also discuss extensively relations with entropic uncertainty relations.

\subsection{Fidelity uncertainty relations}

The main objective of this article is to explore the possibility to quantify uncertainty by yet another notion of proximity between probability measures, namely the classical fidelity (known also in the statistical community as the Bhattacharyya coefficient):
\begin{displaymath}
F(p,q) = \sum_{i=1}^d \sqrt{p(i)}\sqrt{q(i)}
\end{displaymath}
and the related Hellinger distance
\begin{align}\label{eq:Hellinger-distance}
D_H(p,q) = \frac{1}{\sqrt{2}}\Big(\sum_{i=1}^d \Big(\sqrt{p(i)} - \sqrt{q(i)}\Big)^2\Big)^{1/2} = \sqrt{1 - F(p,q)}.
\end{align}

It is known that
\begin{align}\label{eq:distance-comparison}
D_H(p,q)^2 \le D_{TV}(p,q) \le \sqrt{2}D_H(p,q)
\end{align}
and it is not difficult to find examples when $D_{TV}(p,q)$ is indeed of the order $D_H(p,q)^2 \ll D_H(p,q)\ll 1$.

It turns out that the choice of $F$ to quantify the uncertainty has indeed several advantages. Using fidelity allows in particular to
\begin{itemize}
\item obtain uncertainty principles which are expressed in a stronger measure of proximity than those of \cite{MR2932028} and at the same time have a better parameter dependence (in particular one can improve Theorem \ref{th:Fawzi-Hayden-Sen} and obtain optimal dependence of $t$ and $d_B$ on the parameter $\varepsilon$),

\item obtain improved quantitative bounds on the information locking protocols, using stronger proximity measure, allowing for a better dependence on the min-entropy of the message and applicable in a larger range of parameters.

\item obtain, as a byproduct, improved estimates on the dimension of approximately Euclidean subspaces of certain matricial Banach spaces; to our best knowledge such estimates do not follow from known versions of the Dvoretzky theorem.
\end{itemize}

Our approach allows also to study uncertainty when restricted to some proper subsets of states on $\Ha$, e.g. separable states, by connecting the uncertainty to some measures of magnitude related to Gaussian processes, used in asymptotic convex geometry and known to quantify various aspects of geometric behaviour for subsets of high dimensional Euclidean spheres (see \cite{MR2149924,MR2199631,MR2373017}). This is yet another manifestation of the connections between quantum information theory and asymptotic geometric analysis, which in the recent years proved useful in the study of additivity conjectures, random quantum channels or entanglement thresholds (see e.g. \cite{MR2094521,Hastings,MR2605015,MR2802300,MR3139428}).

One possible application of results for proper subsets of quantum states is a result on bit-hiding, presented in Section \ref{sec:further-applications}.

Let us now introduce the uncertainty relations studied in the subsequent part of the paper.  We remark that implicitly they appeared already in Section 2.2. of \cite{MR2932028}, however without optimal existence statements.

\begin{defi}
Let $U_1,\ldots,U_t$ be unitary transformations of $\Ha = \Ha_A\otimes \Ha_B$. For $\varepsilon > 0$ we will say that $\{U_k\}$ satisfy the $\epsilon$-fidelity uncertainty relation on $\Ha_A$ if
\begin{align}\label{eq:fur}
\inf_{|\psi\ra \in S_\Ha} \frac{1}{t} \sum_{k=1}^t F\Big(p^A_{U_k|\psi\ra},\unif([d_A])\Big) \ge 1 - \varepsilon
\end{align}
or equivalently
\begin{align}\label{eq:fur-Hellinger}
\sup_{|\psi\ra\in S_\Ha} \sqrt{\frac{1}{t} \sum_{k=1}^t D_H(p^A_{U_k|\psi\ra},\unif{[d_A]})^2} \le \sqrt{\varepsilon}.
\end{align}
\end{defi}

\begin{remark}
The equivalence between \eqref{eq:fur} and \eqref{eq:fur-Hellinger} follows easily by \eqref{eq:Hellinger-distance}. The form \eqref{eq:fur-Hellinger} combined with the second inequality of \eqref{eq:distance-comparison} shows in particular that
the $\varepsilon$-fidelity uncertainty relation implies the $\sqrt{2\varepsilon}$-metric uncertainty relation (in fact a stronger form of it with the arithmetic mean replaced by the greater quadratic mean).
\end{remark}

In the following sections we present precise formulations of our results, here let us just highlight the most important points.

\begin{itemize}
\item In Theorem \ref{th:fur} of Section \ref{sec:fidelity} we prove that random unitary transformations satisfy the $\varepsilon$-fidelity uncertainty relation for $t \simeq 1/\varepsilon$ and $d_B \simeq \varepsilon$. In particular (by replacing $\varepsilon$ with $\varepsilon^2$) this allows to remove the spurious logarithms in the restriction on $t$ in Theorem \ref{th:Fawzi-Hayden-Sen} and replace the arithmetic mean of total-variation distances by the quadratic mean of greater Hellinger distances.

\item In Section \ref{sec:optimality} we show that $t \gtrsim\frac{1}{\varepsilon^2}$ and $d_B \gtrsim 1/\varepsilon^2$ is necessary for the existence of matrices $U_1,\ldots,U_t$ satisfying the $\varepsilon$-metric uncertainty relations \eqref{eq:mur}. This shows that in the random unitary case the stronger $\varepsilon^2$-fidelity uncertainty appears in the same range of the parameters $t,d_B$ as the weaker $\varepsilon$-total variation uncertainty. It also shows that one cannot use metric uncertainty relations to obtain optimal bounds on entropic uncertainty relation, as in this case the right order of $t$ is $1/\varepsilon$.

\item In Theorem \ref{th:locking}, Section \ref{sec:locking} we show that fidelity uncertainty relations imply locking in Hellinger distance for the protocol introduced in \cite{MR2932028} with parameters improved even with respect to those known in the total variation case. We also show that for Hellinger locking, random unitary transformations give the best possible key length for the locking protocols with uniformly distributed messages and that they almost saturate the locking bounds for the protocol of \cite{MR2932028}.

\item In Section \ref{sec:further-applications} we describe briefly possible applications of the general probabilistic result our proofs are based on (Theorem \ref{th:main}) to the problem of multiple bit hiding.

\item Finally in Section \ref{sec:embeddings} we apply the results of Section \ref{sec:fidelity} to the problem of finding high dimensional Euclidean subspaces of the space $\ell_1^n(\ell_2^m)$, improving known bounds on their dimension.

\end{itemize}

\section{Main results}

\subsection{Notation}
Before we state our main results, let us introduce some additional notation common for all parts of the article.

Recall that $\Ha_A$ and $\Ha_B$ are two complex Hilbert spaces of dimensions $d_A$ and $d_B$ respectively and that by $d = d_Ad_B$ we denote the dimension of $\Ha = \Ha_A\otimes \Ha_B$. Let $\{|a\ra^A\colon a\in[d_A]\}$ and $\{|b\ra^B\colon b \in [d_B]\}$ be distinguished (computational) bases in $\Ha_A$ and $\Ha_B$ respectively. To simplify the notation, most of the time we will simply write $|a\ra$, $|b\ra$ instead of $|a\ra^A$, $|b\ra^B$, however occasionally to stress the underlying tensor structure we will use the superscript.

By $\|\cdot\|$ and $\|\cdot\|_{HS}$ we will denote respectively the operator and Hilbert Schmidt norm of a linear operator on $\Ha$. The notation $\|\cdot\|$ will be also used for the Euclidean norm on $\Ha$ corresponding to the Hilbert structure, i.e. for $|\psi\ra \in \Ha$, $\||\psi\ra\| = \sqrt{\la \psi|\psi\ra}$. Thus for a linear operator $M$, $\|M\| = \sup_{|\psi\ra \in S_\Ha} \|M|\psi\ra\|$ and $\|M\|_{HS} = \sqrt{\sum_{a=1}^{d_A}\sum_{j=1}^{d_B} \| M |a\ra^A |b\ra^B\|^2}$. By $\|\cdot\|_1$ we will denote the trace norm of an operator, i.e. $\|M\|_1 = \tr\sqrt{M^\dag M}$, where $M^\dag$ is the adjoint of $M$.

By $\mathcal{U}_d$ we will denote the unitary group, i.e. the group of all $d\times d$ unitary matrices, equipped with the normalized Haar measure. Modelling an abstract $d$ dimensional Hilbert space $\Ha$ on $\C^d$, we will often identify $\mathcal{U}_d$ with the group $\mathcal{U}_d(\Ha)$ of all unitary transformations of $\Ha$.

Accordingly by a random unitary transformation of $\Ha$ (resp. a random $d \times d$ unitary  matrix) we will always mean a random element of $\mathcal{U}_d(\Ha)$ (resp. $\mathcal{U}_d$) distributed according to the normalized Haar measure.

By $S_\Ha$ we will denote the unit sphere of $\Ha$ and by $S_{\C}^{d-1}$ the unit sphere in $\C^d$. By $S^{d-1}$ we will denote the unit sphere in $\R^d$. Clearly $S_\Ha$, $S_{\C}^{d-1}$ and $S^{2d-1}$ are isometric.

By $\States(\Ha)$ we will denote the set of all (pure and mixed) states on $\Ha$.

By $\log$ we will denote the natural logarithm and by $\log_2$ the logarithm at the base 2. $\Re z$ and $\Im z$ will stand respectively for the real and imaginary part of the complex number $z$, while $|z|$ will denote its modulus.

In the article by $C_1,C_2,\ldots$ we will denote positive numerical constants which are universal, i.e. do not depend on any parameters.

\subsection{Fidelity uncertainty relations \label{sec:fidelity}}
Let us now formulate our main result, which is
\begin{theorem}\label{th:fur}
There exist absolute constants $\Ca,\Cb>0$, with the following property. Let $\varepsilon \in (0,1)$, $d_B \ge \Ca/\varepsilon$, $t \ge \frac{\Ca}{\varepsilon}$ and let $U_1,\ldots,U_{t}$ be i.i.d. random unitary transformations of $\Ha = \Ha_A\otimes \Ha_B$. Then with probability at least $1 - 2e^{-\varepsilon dt/\Cb}$,
the matrices $U_1,\ldots,U_t$ satisfy the $\epsilon$-fidelity uncertainty relation \eqref{eq:fur} on $\Ha_A$. In particular, they also satisfy the $\sqrt{2\varepsilon}$-metric uncertainty relation \eqref{eq:mur} on $\Ha_A$.
\end{theorem}

The proof of the above theorem is based on an adaptation of an argument by Gideon Schechtman used in \cite{MR1008729,MR2199631} to improve the dependence on $\varepsilon$ in the Dvoretzky theorem. The original argument was applied to real Gaussian random matrices, the complex unitary case seems slightly more technical. The key point of Schechtman's idea is to prove that the underlying stochastic process is subgaussian with respect to the Eulidean norm, which allows for an application of Talagrand's generic chaining theorem \cite{MR906527,MR3184689} yielding a comparison of the supremum of this process with the supremum of the canonical Gaussian process $(G_{|\psi\ra})$. Moreover, the comparison holds for general subsets $\Lambda$ of the unit sphere $S_\Ha$, providing an upper bound for
\begin{displaymath}
\E \sup_{|\psi\ra \in \Lambda} \sqrt{1- \frac{1}{t}\sum_{i=1}^t F\Big(p^A_{U_k|\psi\ra},\unif([d_A])\Big)} = \E \sup_{|\psi\ra\in \Lambda} \sqrt{\frac{1}{t}\sum_{k=1}^t  D_H\Big(p_{U_k|\psi\ra}^A,\unif([d_A])\Big)^2}
\end{displaymath}
in terms of $\E \sup_{|\psi\ra \in \Lambda} G_{|\psi\ra}$. Let us mention that the latter quantity has been estimated for many subsets of the sphere important from the geometric or statistical point of view (see e.g. \cite{MR2371614}).

In order to introduce the process $(G_{|\psi\ra}\colon |\psi\ra \in S_\Ha)$ it will be convenient to consider $\Ha$ as a real Hilbert space in the usual way. Let us thus introduce the inner product
\begin{displaymath}
  \la x|y\ra_{\R} = \Re \la x|y\ra
\end{displaymath}
and let $|G\ra$ be a standard Gaussian vector with values in $\Ha$, i.e. a random vector such that for every $|\psi\ra \in \Ha_A$, $\la G|\psi\ra_{\R}$ is distributed according to $\mathcal{N}(0,\||\psi\ra\|^2)$, i.e. has density $g_{|\psi\ra}(x) = \frac{1}{\sqrt{2\pi}\||\psi\ra\|}\exp(-x^2/2\||\psi\ra\|^2)$, where $\||\psi\ra\|^2 = \la \psi|\psi\ra$. The canonical Gaussian process $(G_{\psi})$ on $\Ha$ is defined as
\begin{align}\label{eq:definition-G}
G_{|\psi\ra} = \la G|\psi\ra_{\R}.
\end{align}

Denote also
\begin{align}\label{eq:R-definition}
 R = R_{d_A,d_B} = \E\sqrt{\frac{1}{t}\sum_{k=1}^{t} D_H(p^A_{U_k|\psi\ra},\unif([d_A]))^2}.
\end{align}
Note that since $U_1,\ldots,U_t$ are independent and distributed according to the Haar measure on the unitary group the quantity $R_{d_A,d_B}$ indeed is independent on the choice of $|\psi\ra$ used in \eqref{eq:R-definition} and depends only on the dimensions $d_A,d_B$. In what follows, to simplify the notation we will often make this dependence implicit and write simply $R$ instead of $R_{d_A,d_B}$.

In fact $R$ can be estimated in terms of $d_B$ only. It has been proved in \cite{MR2932028} (Lemma 2.6 therein, see also \cite{MR1491097}) that $\E F(p^A_{U_k|\psi\ra},\unif([d_A])) \ge \sqrt{1 - 1/d_B}$, which via Jensen's inequality gives

\begin{align}\label{eq:M-estimation}
 R & \le \sqrt{\E \frac{1}{t}\sum_{k=1}^{t}  D_H(p^A_{U_k|\psi\ra},\unif([d_A]))^2}=  \Big(1 - \frac{1}{t}\E \sum_{k=1}^t F(p^A_{U_k|\psi\ra},\unif([d_A]))\Big)^{1/2} \nonumber \\
& \le \Big(1 - \sqrt{1-1/d_B}\Big)^{1/2} \le \frac{1}{\sqrt{d_B}}.
\end{align}

The fidelity uncertainty relation of Theorem \ref{th:fur} is a consequence of the following general result.

\begin{theorem}\label{th:main}
For $|\psi\ra \in S_\Ha$ denote
\begin{displaymath}
Y_{|\psi\ra} = \sqrt{\frac{1}{t}\sum_{k=1}^{t} D_H(p^A_{U_k|\psi\ra},\unif([d_A]))^2}.
\end{displaymath}
There exists a universal constant $\Cc$ such that the process $\{Y_{|\psi\ra} \colon |\psi\ra \in S_{\Ha}\}$ is subgaussian with parameter $\frac{\sqrt{\Cc}}{\sqrt{td}}$, i.e. for all $|\psi\ra,|\phi\ra \in S_{\Ha}$,
\begin{align}\label{eq:main-tail}
\p(|Y_{|\psi\ra} - Y_{|\phi\ra}| \ge u) \le 2\exp\Big(-\frac{dtu^2}{\Cc\||\psi\ra-|\phi\ra\|^2}\Big).
\end{align}
Moreover, for every set $\Lambda \subseteq S_{\Ha}$,
\begin{align}\label{eq:main-expectation}
\E \sup_{|\psi\ra,|\phi\ra \in \Lambda} |Y_{|\psi\ra} - Y_{|\phi\ra}| \le \Cd\frac{\E\sup_{|\psi\ra \in \Lambda}  G_{|\psi\ra}}{\sqrt{td}},
\end{align}
where $(G_{|\psi\ra}\colon |\psi\ra \in S_\Ha)$ is the canonical Gaussian process defined by \eqref{eq:definition-G}.
\end{theorem}

Before we pass to the proof of Theorem \ref{th:main} let us demonstrate how to derive from it the fidelity uncertainty principle of Theorem \ref{th:fur}.

\subsubsection{Theorem \ref{th:main} implies Theorem \ref{th:fur}}

Take $\Lambda = S_{\Ha}$ and note that
\begin{displaymath}
\E \sup_{|\psi\ra\in S_\Ha} G_{|\psi\ra} = \E \||G\ra \| \le \sqrt{\E \||G\ra\|^2} = \sqrt{d}.
\end{displaymath}
Therefore, the inequality \eqref{eq:main-expectation} of Theorem \ref{th:main} gives
\begin{displaymath}
\E \sup_{|\psi\ra,|\phi\ra \in S_\Ha} |Y_{|\psi\ra} - Y_{|\phi\ra}| \le \frac{\Cd}{\sqrt{t}}.
\end{displaymath}
Let us now fix an arbitrary point $|\rho\ra \in S_\Ha$. We have
\begin{displaymath}
\E \sup_{|\psi\ra \in S_\Ha} |Y_{|\psi\ra} - Y_{|\rho\ra}| \le \E \sup_{|\psi\ra,|\phi\ra \in S_\Ha} |Y_{|\psi\ra} - Y_{|\phi\ra}| \le \frac{\Cd}{\sqrt{t}}.
\end{displaymath}
By the triangle inequality and \eqref{eq:M-estimation} we further obtain
\begin{align*}
  \E &\sup_{|\psi\ra \in S_\Ha} \sqrt{\frac{1}{t}\sum_{k=1}^{t} D_H(p^A_{U_k|\psi\ra},\unif([d_A]))^2}  = \E \sup_{|\psi\ra\in S_\Ha} Y_{|\psi\ra} \\
  &\le \E Y_{|\rho\ra} + \frac{\Cd}{\sqrt{t}}  = R + \frac{\Cd}{\sqrt{t}} \le \frac{1}{\sqrt{d_B}} + \frac{\Cd}{\sqrt{t}}.
\end{align*}

Thus if $d_B \ge 16/\varepsilon$ and $t \ge 16C_4^2/\varepsilon$ we obtain
\begin{align}\label{eq:expectation-estimate}
  \E \sup_{|\psi\ra \in S_\Ha} \sqrt{\frac{1}{t}\sum_{k=1}^{t} D_H(p^A_{U_k|\psi\ra},\unif([d_A]))^2} \le \sqrt{\varepsilon}/2.
\end{align}

In view of the equivalence between \eqref{eq:fur} and \eqref{eq:fur-Hellinger} this already shows the existence of transformations $U_1,\ldots,U_t$ satisfying the $\varepsilon$-fidelity uncertainty relation. To infer that in fact this relation is satisfied with high probability on the product of unitary groups we need to use the well known concentration of measure inequality, which we recall in the lemma below. The lemma is a well known consequence of the logarithmic Sobolev inequality on the unitary group (see e.g. \cite{MR3109633}) together with the tensorization of entropy and Herbst's argument \cite[Chapter 5]{MR1849347}.

\begin{lemma}\label{le:concentration} Let $\mathcal{U}_d^t$ be the $t$-fold Carthesian product of the unitary group $\mathcal{U}_d$. Let $f \colon \mathcal{U}_d^t \to \R$ be a 1-Lipschitz function with respect to the $\ell_2$ sum of Hilbert-Schmidt distances, i.e.
\begin{displaymath}
|f(V_1,\ldots,V_t) - f(V_1',\ldots,V_t')| \le \sqrt{\sum_{k=1}^t \|V_k - V_k'\|_{HS}^2}
\end{displaymath}
  for all $V_1,\ldots,V_t,V_1',\ldots,V_t \in \mathcal{U}_d$. Let $U_1,\ldots,U_t$ be independent random unitary matrices distributed according to the Haar measure on $\mathcal{U}_d$. Then for any $s \ge 0$,
  \begin{displaymath}
    \p\Big(|f(U_1,\ldots,U_t) - \E f(U_1,\ldots,U_t)| \ge s\Big) \le 2e^{-ds^2/\Ce}.
  \end{displaymath}
\end{lemma}

Thus to complete the proof of the implication [Theorem \ref{th:main} $\Rightarrow$ Theorem \ref{th:fur}] it remains to demonstrate the following lemma.

\begin{lemma}\label{le:Lipschitz-first}
The function $f\colon \mathcal{U}_d^t \to \R$ given by the formula
\begin{displaymath}
f(U_1,\ldots,U_t) = \sup_{|\psi\ra \in S_\Ha} \sqrt{\frac{1}{t}\sum_{k=1}^{t} D_H(p^A_{U_k|\psi\ra},\unif([d_A]))^2}
\end{displaymath}
is $1/\sqrt{2t}$-Lipschitz.
\end{lemma}

Indeed, assuming Lemma \ref{le:Lipschitz-first}, by Lemma \ref{le:concentration} and \eqref{eq:expectation-estimate} we obtain
\begin{align*}
&\p\Big(\sup_{|\psi\ra \in S_\Ha} \sqrt{\frac{1}{t}\sum_{k=1}^{t} D_H(p^A_{U_k|\psi\ra},\unif([d_A]))^2} \le \sqrt{\varepsilon}\Big)
 \\& \ge \p\Big(\sup_{|\psi\ra \in S_\Ha} \sqrt{\frac{1}{t}\sum_{k=1}^{t} D_H(p^A_{U_k|\psi\ra},\unif([d_A]))^2} - \E \sup_{|\psi\ra \in S_\Ha} \sqrt{\frac{1}{t}\sum_{k=1}^{t} D_H(p^A_{U_k|\psi\ra},\unif([d_A]))^2} \le \sqrt{\varepsilon}/2\Big) \\
 &\ge 1 - 2e^{-dt\varepsilon/2\Ce},
\end{align*}
which proves Theorem \ref{th:fur} with $\Ca = 16\max(1,\Cd^2)$, $\Cb = 2\Ce$.

\begin{proof}[Proof of Lemma \ref{le:Lipschitz-first}]
A supremum of $L$-Lipschitz functions is $L$-Lipschitz, so by the triangle inequality in $\ell_2$, it is enough to show that for every $|\psi\ra \in S_\Ha$, the function $U\mapsto D_H(p_{U|\psi\ra}^A,\unif([d_A]))$ is ($1/\sqrt{2}$)-Lipschitz with respect to the Hilbert-Schmidt norm. This is however straightforward as we have
\begin{align*}
&|D_H(p_{U|\psi\ra}^A,\unif([d_A])) - D_H(p_{U'|\psi\ra}^A,\unif([d_A]))| \\
&= \frac{1}{\sqrt{2}}\Big| \sqrt{\sum_{a=1}^{d_A}  \Big(\sqrt{\sum_{b=1}^{d_B}|\la a|\la b|U|\psi\ra|^2}-1/\sqrt{d_A}\Big)^2}
- \sqrt{\sum_{a=1}^{d_A}  \Big(\sqrt{\sum_{b=1}^{d_B}|\la a|\la b|U|\psi\ra|^2}-1/\sqrt{d_A}\Big)^2}\Big|\\
&\le \frac{1}{\sqrt{2}}\sqrt{\sum_{a=1}^{d_A}  \Big(\sqrt{\sum_{b=1}^{d_B}|\la a|\la b|U|\psi\ra|^2} - \sqrt{\sum_{b=1}^{d_B}|\la a|\la b|U'|\psi\ra|^2}\Big)^2}\\
& \le \frac{1}{\sqrt{2}}\sqrt{\sum_{a=1}^{d_A}  \sum_{b=1}^{d_B} |\la a|\la b| U-U'|\psi\ra|^2} = \frac{1}{\sqrt{2}}\| (U - U')|\psi\ra\|  \le \frac{1}{\sqrt{2}}\|(U-U')\|_{HS},
\end{align*}
where the first two inequalities follow from the triangle inequality in $\ell_2$, the second equality from the fact that $\{|a\ra|b\ra\colon a\in [d_A],b\in [d_B]\}$ is an orthonormal basis of $\Ha$ and the last inequality from the comparison between the operator and Hilbert-Schmidt norm of a matrix.
\end{proof}

\subsubsection{Proof of Theorem \ref{th:main}}

Before we start the proof of Theorem \ref{th:main} let us introduce the main probabilistic tools. The first one is a well known lemma, concerning recursive construction of random unitary matrices.

\begin{lemma}[{\cite[Proposition 2.1]{MR2451289}}]\label{le:recursive-representation}
Let $M$ be any $d\times d$ random unitary matrix (not necessarily Haar distributed), whose first column is distributed uniformly on the sphere $S_{\C}^{d-1}$. Let moreover $V$ be a random unitary matrix, distributed according to the Haar measure on $\mathcal{U}_{d-1}$. Assume that $M$ and $V$ are independent. Define the $d\times d$ matrix $W$ with the formula
\begin{displaymath}
  W = \left[\begin{array}{cc}
         1 & 0 \\
         0 & V
       \end{array}\right].
\end{displaymath}
Then the matrix $MW$ is distributed according to the Haar measure on $\mathcal{U}_d$.
\end{lemma}

We will also need the celebrated Majorizing Measure Theorem due to Talagrand \cite{MR906527,MR3184689}, which for our purposes can be formulated as follows.
\begin{theorem}\label{thm:Majorizing-Measures}
Let $(X_{|\psi\ra}\colon |\psi\ra \in \Lambda)$, where $\Lambda \subseteq \Ha$, be a stochastic process, such that for every $|\psi\ra,|\rho\ra \in \Lambda$ and every $u \ge 0$,
\begin{displaymath}
\p(|X_{|\psi\ra} - X_{|\rho\ra}| \ge u) \le 2\exp\Big(-\frac{u^2}{\||\psi\ra - |\rho\ra\|^2}\Big).
\end{displaymath}
Then
\begin{displaymath}
\E \sup_{|\psi\ra,|\rho\ra \in \Lambda} |X_{|\psi\ra} - X_{|\rho\ra}| \le \Cf \E \sup_{|\psi\ra \in \Lambda} G_{|\psi\ra},
\end{displaymath}
where $\Cf$ is an absolute numerical constant.
\end{theorem}

The final ingredient we will use is L\'evy's concentration inequality on the Carthesian product of Euclidean spheres (see e.g. \cite{MR1849347}).

\begin{theorem}\label{thm:concentration-sphere}
Let $f\colon \prod_{i=1}^t S^{n-1} \to \R$ be a 1-Lipschitz function (with respect to the geodesic or Euclidean distance on $\prod_{i=1}^{t}S^{n-1}$) and let $X_1,\ldots,X_t$ be i.i.d. random variables distributed uniformly on $S^{n-1}$. Then for any $u \ge 0$,
\begin{displaymath}
\p(|f(X_1,\ldots,X_t) - \E f(X_1,\ldots,X_k)|\ge u) \le 2\exp\Big(-\frac{(n-1)u^2}{2}\Big).
\end{displaymath}
\end{theorem}

\begin{proof}[Proof of Theorem \ref{th:main} ] Let us start with the proof of \eqref{eq:main-tail}. If $|\psi\ra = e^{\sqrt{-1}s}|\phi\ra$, for some $s \in \R$, there is nothing to prove, since then $Y_{|\psi\ra} = Y_{|\phi\ra}$. Assume thus that $\dim(span(|\psi\ra,|\phi\ra)) = 2$.  Define $|\rho\ra = \frac{1}{2}(|\psi\ra+|\phi\ra) \neq 0$ and fix any orthonormal basis in $\Ha$ with a multiple of $|\rho\ra$ as the first vector. Without loss of generality we can assume that in this basis the transformations $U_k$ are given by matrices $M_kW_k$, where $M_k,W_k$, $k=1,\ldots,t$ are independent random matrices as in Lemma \ref{le:recursive-representation}. Having fixed the basis, in what follows we will identify linear transformations of $\Ha$ with matrices, in particular we will simply write $U_k = M_kW_k$. Denote by $P$ the orthogonal projection on $span(|\rho\ra)$ and let $Q = \id_{\Ha} - P$. Denote $|\chi\ra = \frac{1}{2}(|\psi\ra-|\phi\ra)$, clearly $Q|\chi\ra$ and $|\rho\ra$ are orthogonal.

Note that
\begin{displaymath}
\la \chi|\rho\ra = \frac{1}{4}(\la\psi|\psi\ra + \la\psi|\phi\ra - \la \phi|\psi\ra - \la \phi|\phi\ra) = \frac{1}{2}\Im \la \psi|\phi\ra \in \R \sqrt{-1}.
\end{displaymath}

We have $|\psi \ra = |\rho\ra + |\chi\ra$ and $|\phi\ra = |\rho\ra - |\chi\ra$. Define $\alpha = \frac{\la\rho|\chi\ra}{\la\rho|\rho\ra} \in \R\sqrt{-1}$ and note that $P|\chi\ra = \alpha |\rho\ra$. Thus
\begin{displaymath}
|\psi\ra = (1+\alpha)|\rho\ra + Q|\chi\ra,\quad |\phi\ra = (1-\alpha)|\rho\ra - Q|\chi\ra.
\end{displaymath}

By our choice of the basis it follows that $U_k|\rho\ra = M_k|\rho\ra$ and $U_kQ|\chi\ra = M_k|\chi_k\ra$ where $|\chi_k\ra = W_kQ|\chi\ra$. Thus, the above equalities give

\begin{align}\label{eq:plus-minus}
U_k |\psi\ra = (1+\alpha) M_k|\rho\ra + M_k|\chi_k\ra,\quad U_k|\phi\ra = (1-\alpha)M_k|\rho\ra - M_k|\chi_k\ra.
\end{align}
Moreover, the random vectors $|\chi_k\ra$ are independent and distributed uniformly on the sphere of radius $\| Q|\chi\ra \|$ in the subspace of $\Ha$ orthogonal to $|\rho\ra$. We will denote this sphere by $S$.

For fixed values of the matrices $M_1,\ldots,M_{t}$ and $\beta \in \C$, consider the functions $f_{\beta,k}\colon S \to \R$ defined with the formula
\begin{displaymath}
f_k(|x\ra) = \sqrt{\sum_{a=1}^{d_A} \Big(\sqrt{
\sum_{b=1}^{d_B}\Big|\beta\la a|\la b|M_k|\rho \ra + \la a|\la b|M_k|x\ra\Big|^2}-\frac{1}{\sqrt{d_A}}\Big)^2}
\end{displaymath}
and note that $f_k$ are $1$-Lipschitz. Indeed, for fixed $k$, we have
\begin{align*}
&|f_k(|x\ra) - f_k(|y\ra)| \\
&\le \sqrt{\sum_{a=1}^{d_A} \Big(\sqrt{\sum_{b=1}^{d_B}\Big|\beta \la a|\la b|M_k|\rho \ra + \la a|\la b|M_k|x\ra \Big|^2}-\sqrt{\sum_{b=1}^{d_B}\Big|\beta \la a|\la b|M_k|\rho \ra + \la a|\la b|M_k|y \ra\Big|^2}\Big)^2}\\
&\le \sqrt{\sum_{a=1}^{d_A}\sum_{b=1}^{d_B}\Big|\la a|\la b|M_k(|x\ra -|y\ra) \Big|^2} = \||x\ra - |y\ra\|,
\end{align*}
where the inequalities follow from the triangle inequality in $\ell_2$ and the equality from the fact that $M_k$ is unitary and $\{|a\ra|b\ra\colon a\in[d_A], b\in[d_B]\}$ is an orthonormal basis in $\Ha$.

We have
\begin{align*}
&\sqrt{2}|Y_{|\psi\ra} - Y_{|\phi\ra}| = \sqrt{2}\Big|\sqrt{\frac{1}{t}\sum_{k=1}^{t} D_H(p_{U_k|\psi\ra},\unif([d_A]))^2}- \sqrt{\frac{1}{t}\sum_{k=1}^{t} D_H(p_{U_k|\phi\ra},\unif([d_A]))^2}\Big|\\
& = \frac{1}{\sqrt{t}}\Big|\sqrt{\sum_{k=1}^{t} \sum_{a=1}^{d_A} \Big(\sqrt{\sum_{b=1}^{d_B}|\la a|\la b|U_k|\psi\ra|^2} -\frac{1}{\sqrt{d_A}}\Big)^2}- \sqrt{\sum_{k=1}^{t} \sum_{a=1}^{d_A} \Big(\sqrt{\sum_{b=1}^{d_B}|\la a|\la b|U_k|\phi\ra|^2} -\frac{1}{\sqrt{d_A}}\Big)^2}\\
&= \frac{1}{\sqrt{t}}\Big|\sqrt{\sum_{k=1}^{t} f_{1+\alpha,k}(|\chi_k\ra)^2} - \sqrt{\sum_{k=1}^{t}f_{1-\alpha,k}(-|\chi_k\ra)^2}\Big|,
\end{align*}
where in the third inequality we used \eqref{eq:plus-minus} and the definition of the functions $f_{\beta,k}$. Denote by $\p_W, \E_W$ the probability and integration with respect to the variables $W_k$ and recall that $|\chi_k\ra$ are measurable with respect to these variables, whereas the matrices $M_k$ are independent of them.

Note that
\begin{displaymath}
  f_{1+\alpha,k}(|\chi_k\ra)^2 = \sum_{a=1}^{d_A} \Big(\sqrt{
\sum_{b=1}^{d_B}|1+\alpha|^2 \Big|\la a|\la b|M_k|\rho \ra + \la a|\la b|M_k (1+\alpha)^{-1}|\chi_k\ra\Big|^2}-\frac{1}{\sqrt{d_A}}\Big)^2
\end{displaymath}
and
\begin{displaymath}
  f_{1-\alpha,k}(|\chi_k\ra)^2 = \sum_{a=1}^{d_A} \Big(\sqrt{
\sum_{b=1}^{d_B}|1-\alpha|^2 \Big|\la a|\la b|M_k|\rho \ra + \la a|\la b|M_k (\alpha-1)^{-1}|\chi_k\ra\Big|^2}-\frac{1}{\sqrt{d_A}}\Big)^2.
\end{displaymath}

Since $\alpha \in \R\sqrt{-1}$, we have $|1+\alpha| = |1-\alpha|$. Moreover, the distribution of $|\chi_k\ra$ is uniform on $S$, which implies that $(1+\alpha)^{-1}|\chi_k\ra$
and $(\alpha-1)^{-1}|\chi_k\ra$ have the same distribution (uniform on $|1+\alpha|^{-1}S$). Together with the independence of $|\chi_k\ra$, $k=1,\ldots,t$, this implies that conditionally on the random matrices $M_1,\ldots,M_{t}$, the random variables
\begin{displaymath}
  \sqrt{\sum_{k=1}^{t} f_{1+\alpha,k}(|\chi_k\ra)^2} \textrm{  and  }\sqrt{\sum_{k=1}^{t}f_{1-\alpha,k}(-|\chi_k\ra)^2}
\end{displaymath}
have the same distribution.
This allows us to write
\begin{align*}
\sqrt{2}|Y_{|\psi\ra} - Y_{|\phi\ra}| \le& \frac{1}{\sqrt{t}}\Big|\sqrt{\sum_{k=1}^{t} f_{1+\alpha,k}(|\chi_k\ra)^2} - \E_W \sqrt{\sum_{k=1}^{t} f_{1+\alpha,k}(|\chi_k\ra)^2}\Big|\\
&+ \frac{1}{\sqrt{t}}\Big|\sqrt{\sum_{k=1}^{t} f_{1-\alpha,k}(|-\chi_k\ra)^2} - \E_W \sqrt{\sum_{k=1}^{t} f_{1-\alpha,k}(|-\chi_k\ra)^2}\Big|.
\end{align*}
As a consequence
\begin{align*}
\p_W(|Y_{|\psi\ra} - Y_{|\phi\ra}| \ge u) \le 2\p_W\Big(\frac{1}{\sqrt{t}}\Big|\sqrt{\sum_{k=1}^{t} f_{1+\alpha,k}(|\chi_k\ra)^2} - \E_W \sqrt{\sum_{k=1}^{t} f_{1+\alpha,k}(|\chi_k\ra)^2}\Big| \ge \frac{u\sqrt{2}}{2}\Big)
\end{align*}

By the 1-Lipschitz property of $f_{1+\alpha,k}$ and the triangle inequality in $\ell_2$ it follows that the function
\begin{displaymath}
(|\chi_1\ra,\ldots,|\chi_{t}\ra) \mapsto \frac{1}{\sqrt{t}}\sqrt{\sum_{k=1}^{t} f_{1+\alpha,k}(|\chi_k\ra)^2}
\end{displaymath}
is $\frac{1}{\sqrt{t}}$-Lipschitz on $\prod_{k=1}^{t}S$. Therefore, by the concentration inequality on the product of unit spheres, Theorem \ref{thm:concentration-sphere}, (note that $S$ is of real dimension $n=2(d-1)-1 = 2d-3$), we get
\begin{displaymath}
\p_W(|Y_{|\psi\ra} - Y_{|\phi\ra}| \ge u) \le 4\exp\Big(-\frac{(2d-3)tu^2 }{4\|Q|\chi\ra\|^2}\Big)\le 2\exp\Big(-\frac{dtu^2}{\Cc\||\psi\ra-|\phi\ra\|^2}\Big),
\end{displaymath}
where in the last inequality we used the fact that $Q$ is a contraction, the definition of $|\chi\ra$ and we have adjusted the constants.
Integrating the above inequality with respect to the random matrices $M_k$, together with the Fubini theorem proves \eqref{eq:main-tail}. The inequality \eqref{eq:main-expectation} follows from \eqref{eq:main-tail} by Theorem \ref{thm:Majorizing-Measures}.
\end{proof}

\subsubsection{Uncertainty over a subset of states}

We would like to point out that Theorem \ref{th:main} can be used to quantify uncertainty uniformly over proper subsets of the set of all pure states. Basically the same proof as in the implication [Theorem \ref{th:main} $\implies$ Theorem \ref{th:fur}] gives us the following result, which will be used in Section \ref{sec:further-applications}.

\begin{theorem}\label{thm:arbitrary-set}
If $\Lambda \subseteq S_\Ha$ and $d_B \ge \frac{16}{\varepsilon}$ then with probability at least $1 - 2\exp(-\varepsilon dt/\Cb)$ for all $|\psi\ra \in \Lambda$,
\begin{displaymath}
\sqrt{\frac{1}{t}\sum_{k=1}^t D_H(p_{U_k|\psi\ra}^A,\unif([d_A]))^2} \le 3\sqrt{\varepsilon}/4 + \Cd\frac{\E\sup_{|\psi\ra \in \Lambda}  G_{|\psi\ra}}{\sqrt{td}}.
\end{displaymath}
\end{theorem}

\begin{example}\label{ex:separable}
For instance if $\Lambda = S_{\Ha_A}\otimes S_{\Ha_B}$ is the set of all separable states then, as one can easily see, the quantity $2^{-1/2}\E\sup_{|\psi\ra \in \Lambda}  G_{|\psi\ra}$ is the expected operator norm of a random $d_A \times d_B$ matrix with i.i.d. standard complex Gaussian coefficients, which, as is well known (see e.g. \cite{MR520217,MR2760897}) is of the order $\sqrt{d_A} + \sqrt{d_B}$. Thus
\begin{displaymath}
  \Cd\frac{\E\sup_{|\psi\ra \in \Lambda}  G_{|\psi\ra}}{\sqrt{td}} \le \Co\Big(\frac{1}{\sqrt{td_A}} + \frac{1}{\sqrt{td_B}}\Big).
\end{displaymath}
In particular, if $d_A,d_B \ge 64 \Co^2/\varepsilon^2$, by setting $t = 1$ in the above theorem (applied with $\varepsilon^2$ instead of $\varepsilon$), we obtain that with probability at least $1 - 2\exp(-\varepsilon^2 d/\Cb)$,
\begin{align}\label{eq:uncertainty-separable-single-t}
D_H(p_{U_1|\psi\ra}^A,\unif([d_A])) \le \varepsilon.
\end{align}
Thus to provide uniform uncertainty relations for the set of separable states, it is enough to use a single unitary matrix. In general, the order of magnitude of $t$ required for providing an uncertainty relation uniformly over some set $\Lambda$ of states is governed by the Gaussian magnitude of this set, which resembles phenomena known from the theory of dimension reduction in asymptotic convex geometry or compressed sensing (see e.g. \cite{MR2149924,MR2199631,MR2373017}).
\end{example}

\subsection{Optimality. Lower bounds for metric uncertainty relations \label{sec:optimality}}

 We will now address the question of necessary conditions for the existence of transformations $U_1,\ldots,U_t$, satisfying $\varepsilon$-metric uncertainty principles \eqref{eq:mur} of Fawzi-Hayden-Sen. We will show that they cannot exist unless $t, d_B \gtrsim C/\varepsilon^2$. Since $\varepsilon$-metric uncertainty relations are weaker than $\varepsilon^2/2$-fidelity uncertainty relations, this shows in particular that Theorem \ref{th:fur} is optimal up to the value of the universal constant $\Ca$. At the same time it shows that random unitary matrices satisfy the stronger fidelity uncertainty relations at roughly the same range of parameters as the weaker metric uncertainty relations.

\begin{theorem}
\label{thm:lower_bound}
If $d_A>1$ and $U_1,\ldots,U_{t}$ are unitary transformations of $\Ha$, satisfying the $\varepsilon$-metric uncertainty relation \eqref{eq:mur} on $\Ha_A$, then $d_B,t \ge \frac{1}{\Cg\varepsilon^2}$, where $\Cg$ is an absolute constant.
\end{theorem}

\begin{remark}
We would like to stress that in the above theorem we do not impose any random model on the matrices $U_i$, they can be arbitrary deterministic matrices.
\end{remark}

Before we prove Theorem \ref{thm:lower_bound}, let us recall that a measure $\mu$ on $\R^d$ is called log-concave if for all nonempty compact subsets $A,B\subseteq \R^d$ and all $\lambda \in (0,1)$,
\begin{displaymath}
  \mu(\lambda A + (1-\lambda)B) \ge \mu(A)^\lambda \mu(B)^{1-\lambda}.
\end{displaymath}
We refer the Reader to \cite[Chapter 2]{MR3185453} for a comprehensive description of this important class of measures. In our argument we will need their two properties. The first one is an easy consequence of the Pr\'ekopa-Leindler inequality (see e.g. \cite[Theorem 1.2.3]{MR3185453}) and asserts that uniform distributions on convex sets are log-concave. The second one is a special case of moment comparison for log-concave distributions due to Borell (see e.g. \cite[Appendix III]{MR856576}).

\begin{lemma}\label{le:Borell}
There exists a universal constant $\Ch$ such that for any real random variable $X$ with a log-concave distribution,
\begin{displaymath}
(\E |X|^2)^{1/2} \le \Ch \E|X|.
\end{displaymath}
\end{lemma}

\begin{proof}[Proof of Theorem \ref{thm:lower_bound}]
Assume that $U_1,\ldots,U_{t}$ satisfy \eqref{eq:mur}. To show lower bounds on $t$ and $d_B$ we will use the probabilistic method. Let $|\psi\ra$ be a random state, distributed uniformly on the unit sphere in $\Ha$, so that $\la a|^A\la b|^B U_k|\psi\ra$ are distributed uniformly on the unit sphere in $\C^d$. Denote $X^{(k)}_{a,b} = |\la a|^A\la b|^BU_k|\psi\ra|^2$. It is well known that the random vector $(X^{(k)}_{a,b})_{a \in [d_A],b\in[d_B]}$ is distributed uniformly on the simplex $\Delta_{d-1} = \{x = (x_1,\ldots,x_d) \colon \sum_{i=1}^d x_i = 1, x_i \ge 0\}$.

The uniform distribution on $\Delta_{d-1}$ is log-concave. Since, as one can easily see, linear images of log-concave measures are log-concave, Lemma \ref{le:Borell} implies that
\begin{displaymath}
\E |\sum_{b=1}^{d_B} X^{(k)}_{a,b} - \frac{1}{d_A}| \ge \frac{1}{\Ch}(\E |\sum_{b=1}^{d_B} X^{(k)}_{a,b} - \frac{1}{d_A}|^2)^{1/2}.
\end{displaymath}
Using \eqref{eq:mur} we thus obtain
\begin{align}\label{eq:expected-distance-lower-bound}
2\varepsilon \ge 2\E \frac{1}{t}\sum_{k=1}^{t} D_{TV}(p^A_{U_k|\psi\ra},\unif([d_A])) = d_A \E |\sum_{b=1}^{d_B} X^{(1)}_{a,b} - \frac{1}{d_A}| \ge \Ch^{-1}d_A\sqrt{\E |\sum_{b=1}^{d_B} X^{(1)}_{a,b} - \frac{1}{d_A}|^2}.
\end{align}

A direct calculation shows that
\begin{displaymath}
\E X^{(1)}_{a,b} = \frac{1}{d}, \; \Var(X^{(1)}_{a,b}) = \frac{d-1}{d^2(d+1)}
\end{displaymath}
and for $ (a,b) \neq (a',b')$,
\begin{displaymath}
\Cov(X_{a,b}^{(1)},X_{a',b'}^{(1)}) = -\frac{1}{d^2(d+1)}.
\end{displaymath}
Therefore, the right hand side of \eqref{eq:expected-distance-lower-bound} equals (recall that $d=d_Ad_B$)
\begin{displaymath}
d_A \Ch^{-1} \sqrt{d_B \frac{d-1}{d^2(d+1)} - d_B (d_B-1)\frac{1}{d^2(d+1)}} = \Ch^{-1}\sqrt{\frac{(d_A-1)}{d_Ad_B+1}}\ge \frac{1}{2\Ch\sqrt{d_B}},
\end{displaymath}
which gives $d_B \ge \frac{1}{16\Ch^2\varepsilon^2}$ (we used here that $d_A\ge 2$) and proves the desired lower bound on $d_B$.

Let us now pass to the lower bound on $t$. Define the projection operators
\begin{displaymath}
A_{ka} = \sum_{b=1}^{d_B} U_k^\dag(|a\ra|b\ra\la a|\la b|)U_k.
\end{displaymath}
We have
\begin{align*}
\frac{2}{t}\sup_{|\psi\ra\in S_\Ha} \sum_{k=1}^{t} D_{TV}(p^A_{U_k|\psi\ra},\unif([d_A]))
&= \frac{1}{t}\sup_{|\psi\ra\in S_\Ha}\sup_{\alpha \in \{-1,1\}^{[t]\times[d_A]}} \sum_{k=1}^{t}\sum_{a=1}^{d_A}\alpha_{ka} \la \psi| (A_{ka}- \frac{1}{d_A} \id_\Ha)|\psi\ra\\
&= \frac{1}{t}\sup_{\alpha \in \{-1,1\}^{[t]\times[d_A]}} \Big\|\sum_{k=1}^{t}\sum_{a=1}^{d_A} \alpha_{ka} (A_{ka}- \frac{1}{d_A} \id_\Ha)\Big\|\\
& \ge \frac{1}{t\sqrt{d}}\sup_{\alpha \in \{-1,1\}^{[t]\times[d_A]}} \Big\|\sum_{k=1}^{t}\sum_{a=1}^{d_A} \alpha_{ka}  (A_{ka}- \frac{1}{d_A} \id_\Ha)\Big\|_{HS}.
\end{align*}

Let now $\varepsilon_{ka}$, $k=1,\ldots,t$, $a = 1,\ldots,d_A$ be i.i.d. Rademacher random variables, i.e. $\p(\varepsilon_{ka} = 1) = \p(\varepsilon_{ka} = -1) = 1/2$. From the above inequality we readily obtain
\begin{align}\label{eq:intermediate}
\frac{2}{t}\sup_{|\psi\ra\in S_\Ha} \sum_{k=1}^{t} D_{TV}(p^A_{U_k|\psi\ra},\unif([d_A])) \ge \frac{1}{t\sqrt{d}}\E \Big\|\sum_{k=1}^{t}\sum_{a=1}^{d_A} \varepsilon_{ka}(A_{ka}- \frac{1}{d_A} \id_\Ha)\Big\|_{HS}.
\end{align}

By the classical Khintchine-Kahane inequality (see e.g. \cite{MR1102015}, see also \cite{MR1267715} for the optimal constant), the right hand side above is bounded from below by
\begin{align*}
\frac{1}{t\sqrt{2d}}\sqrt{\sum_{k=1}^{t}\sum_{a=1}^{d_A} \|A_{ka}- \frac{1}{d_A} \id\|_{HS}^2} = \frac{1}{t\sqrt{2d_Ad_B}}\sqrt{td_A\Big((d_B(1-\frac{1}{d_A})^2 + (d_Ad_B - d_B)\frac{1}{d_A^2}\Big)},
\end{align*}
with the equality following from the fact that $A_{ka}$ are orthogonal projections on subspaces of dimension $d_B$. Using the assumption $d_A \ge 2$, one can easily see that the last expression is bounded from below by $\frac{1}{2\sqrt{2t}}$. Since by the assumption \eqref{eq:mur} and \eqref{eq:intermediate} it is also bounded from above by $2\varepsilon$, we get
\begin{displaymath}
t \ge \frac{1}{32\varepsilon^2},
\end{displaymath}
which ends the proof of Theorem \ref{thm:lower_bound}.
\end{proof}

\subsection{Information locking \label{sec:locking}}
In this section we will discuss applications of Theorem \ref{th:fur} to locking of classical information in quantum states. The locking phenomenon was first described in \cite{PhysRevLett.92.067902} and was expressed there in terms of mutual information between the outcomes of measurements performed on a bipartite quantum state. Later in \cite{MR2932028} the definition of locking was strengthened to a form involving the total variation distance. Before we state a simple generalization of the definition from \cite{MR2932028} let us provide a brief informal description of the locking phenomenon. In short, information locking occurs when given a large bipartite state shared between two parties (say Alice and Bob), Bob cannot obtain almost any information about the Alice part of the state by performing a local measurement on his part, however when sent a small number of classical bits by Alice he obtains full information about her part. One says then that a certain amount of information was locked in the bipartite quantum state shared by Alice and Bob and was unlocked with just a small number of bits.

The strength of locking is described by quantitative relations between the size of the system, information Bob could obtain initially and the number of bits that need to be sent in order to provide Bob with full information. Clearly, there is a lot of flexibility in choosing the characteristics that measure the size of the system  and initial information/uncertainty of Bob, which leads to various formalizations of the locking phenomenon. Below we state an abstract definition which encompasses the one from \cite{MR2932028}. It will also allow us to express locking in terms of the Hellinger distance \eqref{eq:Hellinger-distance}, which will lead to a strengthening of results from \cite{MR2932028}.

We remark that over the years following its introduction information locking has found multiple applications in quantum information theory, quantum cryptography and even in physics of black holes. Describing them is beyond the scope of this paper, therefore we refer the Readers to the articles \cite{PhysRevLett.92.067902,MR2932028,Locking-decoding,1742-6596-143-1-012008}.

In what follows we will deal with two kinds of probabilities. The probability related to the random constructions involving random unitary matrices (distributed as before according to the Haar measure on the unitary group), and probabilities related to intrinsic randomness in the results of quantum measurements. To avoid confusion, we will denote the probabilities related to Haar unitary matrices by $\p$ and probabilities related to the quantum theory by $\pp$.

Recall that for a probability vector $p = (p(1),\ldots,p(N))$ the min entropy $H_{\min}(p)$ is defined as $H_{\min}(p) = - \log_2 (\max_{i\le N} p(i))$.
Recall also that by $\States(\Ha)$ we denote the set of all states on the Hilbert space $\Ha$.

\begin{defi}\label{def:metric-locking} Let $n$ be a positive integer and let $\mathcal{P}$ be the set of all probability distributions on $[2^n]$. Let further $\rho\colon \mathcal{P}\times\mathcal{P} \to [0,\infty)$ be an arbitrary measure of proximity between probability distributions, such that $\rho(p,q) = 0$ if and only if $p=q$. Recall that by $\|\cdot\|_1$ we denote the trace norm on the space of operators acting on $\Ha$.

Consider any $l \in [n]$, and $\varepsilon \ge 0$. An encoding $\mathcal{E} \colon [2^n]\times [t] \to \States(\Ha)$ is said to be $(\rho,l,\varepsilon)$-locking for the quantum system $\Ha$ if the following two conditions are satisfied:
\begin{itemize}
\item[a)] For all $x\neq x' \in [2^n]$ and all $k\in[t]$,
\begin{align}\label{eq:identification}
\frac{1}{2}\|\mathcal{E}(x,k) - \mathcal{E}(x',k))\|_1 = 1.
\end{align}
\item[b)] Let $X$ (the message) be a random variable on $[2^n]$ with min-entropy $H_{\min}(X) \ge l$ and $K$ (the key) be an independent uniform random variable on $[t]$. For any measurement $\{M_i\}$ on $\Ha$ and any outcome $i$ (by $I$ we denote the corresponding random variable),
    \begin{align}\label{eq:locking-rho}
    \rho(\pp(X \in \cdot|I = i), \pp(X\in \cdot)) \le \varepsilon.
    \end{align}
\end{itemize}
If $l = n$ (i.e. if $X$ is distributed uniformly on $[2^n]$) we speak simply about $\varepsilon$-locking.
\end{defi}

In particular if $\rho = D_{TV}$ we will speak about total variation locking and when $\rho = D_H$ about Hellinger locking. It follows from \eqref{eq:distance-comparison} that if $\mathcal{E}$ is $(D_{H},l,\varepsilon)$-locking  then it is $(D_{TV},l,\sqrt{2}\varepsilon)$-locking

The intuitive meaning of the above definition is clear. Specializing to the case of uniform messages, if one uses $\log_2 t$ random bits independent of the original message to create the key $K$ and encodes the message $X$ together with $K$ in $\mathcal{E}(X,K)$, then no measurement on the system in state $\mathcal{E}(X,K)$ can reveal any significant information about $X$ (as the posterior distribution $\pp(X \in \cdot|I = i)$ is close to uniform). On the other hand, knowing that $K = k$ allows to identify the message $X$, since thanks to \eqref{eq:identification} the states $\mathcal{E}(x,k)$, $x\in[2^n]$, have pairwise orthogonal supports.

The following theorem concerning the relation between metric uncertainty principles \eqref{eq:mur} and locking in total variation distance has been proved in \cite{MR2932028}.

\begin{theorem}[Fawzi-Hayden-Sen, \cite{MR2932028}]\label{thm:metric-locking}
Let $\varepsilon \in (0,1)$ and let $U_1,\ldots,U_{t}$ be unitary transformations of $\Ha$, which satisfy the $\varepsilon$-metric uncertainty principle \eqref{eq:mur} on $\Ha_A$. Assume that $d_A = 2^n$ and define $\mathcal{E}\colon [2^n]\times [t] \to \States(\Ha)$ as
\begin{align}\label{eq:locking-scheme}
\mathcal{E}(x,k) = \frac{1}{d_B}\sum_{b=1}^{d_B}U_k^\dag(|x\ra^A\la x|\otimes |b\ra^B\la b|)U_k.
\end{align}
Then $\mathcal{E}$ is $(D_{TV},\varepsilon)$-locking. Moreover for all $l\in [0,n]$ such that $2^{l-n} > \varepsilon$, it is $(D_{TV},l,\frac{2\varepsilon}{2^{l-n}-\varepsilon})$-locking.
\end{theorem}

Our main result concerning locking is the following theorem, which provides the relation between fidelity uncertainty relations \eqref{eq:fur} and locking in Hellinger distance.

\begin{theorem}\label{th:locking}
Let $\varepsilon \in (0,1)$ and let $U_1,\ldots,U_{t}$ be unitary transformations of $\Ha$, which satisfy the $\varepsilon^2$-fidelity uncertainty principle \eqref{eq:fur}. Assume that $d_A = 2^n$ and define $\mathcal{E}\colon [2^n]\times [t] \to \States(\Ha)$ by \eqref{eq:locking-scheme}.
Then $\mathcal{E}$ is $\varepsilon$-locking in Hellinger distance. Moreover for all $l\in [0,n]$ such that $ 2^{l-n} > 2\varepsilon^2$, it is $(D_H,l,\frac{2\varepsilon}{ 2^{(l-n)/2}-\sqrt{2}\varepsilon})$-locking.
\end{theorem}

In combination with Theorem \ref{th:fur} we immediately obtain the following corollary.

\begin{corollary}\label{cor:locking}
There exist $(D_H,\varepsilon)$-locking schemes encoding an $n$-bit uniform message into at most $n+ 2\log_2(1/\varepsilon) + O(1)$ qubits and using a key of length at most $2\log_2(1/\varepsilon) + O(1)$ bits.
\end{corollary}

We postpone the proof of the theorem to Section \ref{sec:proof-locking} to discuss first its optimality and relation with previous results.

In \cite{MR2932028} the Authors by means of Theorems \ref{th:Fawzi-Hayden-Sen} and \ref{thm:metric-locking} show the existence of $(D_{TV},\varepsilon)$-locking schemes encoding $n$-bit uniform messages into $n + 2\log_2(1/\varepsilon) + O(1)$ qubits using a key of $2\log_2(1/\varepsilon) + O(\log\log(1/\varepsilon))$ bits. They also point out that the key length of any $(D_{TV},\varepsilon)$-locking protocol must be at least $\log_2(1/(\varepsilon+2^{-n})) = \log_2(1/\varepsilon) - o_\varepsilon(1)$ as $n\to \infty$. Theorem \ref{th:locking} improves upon the result in \cite{MR2932028} as it considers a stronger proximity measure $D_H$ and eliminates the $\log\log$ part in bound on the the key length. Clearly, in view of the aforementioned lower bound, it would be of interest to eliminate the factor of 2 in the result of \cite{MR2932028}. However it turns out that for locking in Hellinger distance this factor is in fact necessary, which shows that Corollary \ref{cor:locking} provides key length which is optimal up to an absolute additive constant. This is formalized in the following proposition.

\begin{prop}\label{prop:locking-optimality-general} Any $(D_H,\varepsilon)$-locking protocol encoding $n$-bit messages has key length at least $\log_2(\frac{1}{2\varepsilon^2+ 2^{1-n/2}}) \ge 2\log_2(1/\varepsilon) - 1 - o_\varepsilon(1)$ bits.
\end{prop}

\begin{proof}
Consider the encoding procedure under the assumption that $K=1$, which exists by \eqref{eq:identification}. Thus $I$ takes values in $[2^n]$. By the decoding condition we have $\pp(I = X|K=1) = 1$. Thus
\begin{displaymath}
\frac{1}{t} = \pp(K = 1) \le \p(I=X)
\end{displaymath}
and as a consequence
\begin{displaymath}
\sum_{i=1}^{2^n} \pp(X=i|I=i)\p(I=i) = \pp(X=I) \ge 1/t,
\end{displaymath}
which implies that there exists $i \in [2^n]$ such that $\pp(X=i|I=i) \ge 1/t$. Fix any such $i$ and denote $p = \pp(X=i|I=i)$.

By \eqref{eq:locking-rho} and \eqref{eq:Hellinger-distance} we have
\begin{displaymath}
\sum_{x=1}^{2^n} \sqrt{\pp(X=x|I=i)}2^{-n/2} \ge 1 -\varepsilon^2.
\end{displaymath}
By the Cauchy-Schwarz inequality, we get
\begin{displaymath}
\sqrt{p2^{-n}} + \sqrt{\frac{2^n-1}{2^n}}\sqrt{1-p} \ge 1 - \varepsilon^2,
\end{displaymath}
i.e.
\begin{displaymath}
p2^{-n} + (1-p)(1-2^{-n}) +2\sqrt{p(1-p)2^{-n}(1-2^{-n})} \ge 1 -2\varepsilon^2 + \varepsilon^4.
\end{displaymath}
Using the fact that $p \le 1$, $p(1-p) \le 1/4$, we get $2\varepsilon^2 + 2^{1-n/2} \ge p \ge t^{-1}$, which yields
\begin{displaymath}
t \ge \frac{1}{2\varepsilon^2 + 2^{1-n/2}},
\end{displaymath}
ending the proof.
\end{proof}

Let us also discuss optimality of the second part of Theorem \ref{th:locking}, concerning the restrictions on the min-entropy of the message $X$. This time we will not do it for arbitrary locking schemes but just for the locking scheme \eqref{eq:locking-scheme}.  Substituting $r := 2^{l-n}/(2\varepsilon^2)$, we see that for $r > 1$, by Theorem \ref{th:locking} $\mathcal{E}$ is $(D_H,l,\frac{\sqrt{2}}{\sqrt{r} - 1})$ locking. In other words, to get a meaningful bound from Theorem \ref{th:locking}, we need
\begin{displaymath}
l \ge n - 2\log_2(1/\varepsilon) + \log_2 r,
\end{displaymath}
with $r > \sqrt{2}+1$. Note that this allows for a greater range of $l$ than Theorem \ref{thm:metric-locking}, which requires $l$ to be larger than $n - \log_2(1/\varepsilon)$.

It turns out that for the map $\mathcal{E}$ under $\varepsilon^2$-fidelity uncertainty principle assumptions, the barier $n - 2\log_2(1/\varepsilon)$ for $l$ is essentially optimal (up to universal additive constants)  for locking both in the total variation and Hellinger distance. This is shown in the following proposition.

\begin{prop}\label{prop:locking-optimality-random}
There exists $\Cj$ such that for any $n$ there exist unitary transformations $U_1,\ldots,U_t$ satisfying the $\varepsilon^2$-fidelity uncertainty principle and such that for any integer $l < n - 2\log_2(1/\varepsilon) - \Cj$ and any message $X$ distributed uniformly on a set $S\subseteq [2^n]$ of cardinality $2^l$, there exists a measurement $\{M_{fail}\}\cup \{M_{x,k}\}_{x\in S,k\in [t]}$ such that for all $x \in S, K \in [t]$,
\begin{align}\label{eq:conditional}
\pp(X = x,K=k|I = (x,k)) = 1.
\end{align}
\end{prop}

\begin{remark} The above proposition shows that if the result of the measurement is $(x,k)$, we are able to perfectly identify the message and the key.
Note that if $l \to \infty$ with $n$, we obtain
\begin{displaymath}
  D_{TV}(\p(X\in \cdot|I = (x,k)),\p(X\in \cdot)) = \frac{1}{2}(|1 - 2^{-l}| + (2^{l}-1)|0 - 2^{-l}|) = 1 - 2^{-l} \to 1.
\end{displaymath}
The result \emph{fail} of the measurement happens with probability strictly between 0 and 1. Thus the assertion of the proposition should not be confused with unconditional decodability of the message, in other words it may be still possible that the (unconditional) probability that one recovers $X$ is small.
\end{remark}

\begin{proof}[Proof of Proposition \ref{prop:locking-optimality-random}]

Let $d_A = 2^n$, $d_B = \lceil \Ca/\varepsilon^2\rceil$. Define $\mathcal{E}$ via \eqref{eq:locking-scheme} where $U_k$ are random unitary transformations and $t= \lfloor \Cl/\varepsilon^2\rfloor$ for a sufficiently large constant $\Cl$.  Let now $X,K$ be independent random variables, independent of the matrices $U_k$ and such that $K$ is distributed uniformly on $[t]$ and $X$ distributed uniformly on a set $S \subseteq [2^n]$ of cardinality $2^l$ with $l \le n - 2\log_2(1/\varepsilon) - \Cj$, where $\Cj$ is a universal constant which will be fixed later on. We will now prove that with probability one on the unitary group there exists a measurement satisfying \eqref{eq:conditional}. Since by Theorem \ref{th:fur}, with positive probability the matrices $U_1,\ldots,U_t$ satisfy the $\varepsilon^2$-fidelity uncertainty relation this will end the proof of the proposition.

The existence of the measurement follows just from a dimension count argument.
Fix $x\in S$ and $k\in[t]$ and consider the subspaces $H_{x,k,y,i} := span(\{U_i(|y\ra|b\ra)\colon b \in [d_B]\}$ for $(y,i) \neq (x,k)$. Each of them is of dimension $d_B$, so the space
$H_{x,k} = \oplus_{y \in S, i\in[t]\colon (y,l) \neq (x,k)} H_{x,k,y,i}$ is of dimension at most $t |S| d_B \le \Cl \varepsilon^{-2} 2^l d_B$. Thus if
$l < n - 2\log_2(1/\varepsilon) - \Cj$, where $\Cj = \log_2 \Cl$, then $\dim H_{x,k} < 2^nd_B = \dim \Ha$ and so there exists a state $|e_{x,k}\ra\in \Ha$ orthogonal to $H_{x,k}$. Note that since $U_k$ is independent of $\{U_i\}_{i\neq k}$ and the conditional distribution of $U_k(|x\ra|1\ra)$ given  $\{U_k(|y\ra|b\ra)\}_{y \in S, y\neq x, b \in [d_B]}$ is the uniform distribution on the unit sphere in the orthogonal complement of $span\{U_k(|y\ra|b\ra)\colon y \in S, y\neq x, b\in[d_B]\}$, we obtain that with probability one on the product of unitary groups, $\la e_{x,k}|U_k|x\ra|1\ra \neq 0$ (we use here that $|e_{x,k}\ra$ can be chosen in a measurable way with respect to $\{U_i(|y\ra|b\ra)\colon b \in [d_B], (y,i)\neq (x,k)\}$).

Let us now set $M_{x,k} = \frac{1}{|S|t}|e_{x,k}\ra\la e_{x,k}|$. Then $\sum_{x\in S,k\in [t]} M_{x,k}$ is a non-negative definite operator of norm at most one, so by the spectral theorem we can find a non-negative definite operator $M_{fail}$ such that $M_{fail} + \sum_{x,k} M_{x,k} = \id$ (i.e. the family $\{M_{fail}\}\cup \{M_{x,k}\}_{x\in S,k\in [t]}$ is a POVM). Now, for $x \in S$, $k \in [t]$ and $(y,i)\neq (x,k)$ we have
\begin{align*}
\pp(I = (x,k)|X = y, K = i) &= \frac{1}{d_B}\sum_{b=1}^{d_B}\tr (U_i(|y\ra^A\la y|\otimes |b\ra^B\la b|)U_i^\dag M_{x,k}) \\
& = \frac{1}{|S|td_B}\sum_{b=1}^{d_B} |\la e_{x,k}|U_i|y\ra|b\ra|^2 = 0
\end{align*}
by the construction of the vector $|e_{x,k}\ra$.  Moreover,
\begin{align*}
  \pp(I = (x,k)) &\ge \frac{1}{|S|t} \p(I = (x,k)|X = x, K = k) = \frac{1}{|S|td_B} \tr (\sum_{b=1}^{d_B} U_k(|x\ra^A\la x|\otimes |b\ra^B\la b|)U_k^\dag M_{x,k})\\
  & \ge \frac{1}{|S|td_B} |\la e_{x,k}|U_k|x\ra |1\ra|^2 > 0.
\end{align*}
Thus \eqref{eq:conditional} follows by the Bayes rule.
\end{proof}

\subsubsection{Proof of Theorem \ref{th:locking} \label{sec:proof-locking}}

The proof follows the ideas used in \cite{MR2932028} to demonstrate Theorem \ref{thm:metric-locking}.

\begin{proof}[Proof of Theorem \ref{th:locking}]

Note that thanks to \eqref{eq:Hellinger-distance}, for $\rho = D_{H}$, the inequality \eqref{eq:locking-rho} is equivalent to
\begin{align}\label{eq:locking-fidelity}
\sum_{x = 1}^{2^n} \sqrt{\pp(X = x|I=i)}\sqrt{\pp(X=x)} \ge 1 - \varepsilon^2.
\end{align}
By concavity of the square root and the spectral theorem it follows that without loss of generality we may assume that each $M_i$ is of rank one. Indeed, if $M_i = \sum_{j} \xi_{i,j}|e_{ij}\ra\la e_{ij}|$ and $J$ describes the result of the fine-grained measurement related to the POVM $\{\xi_{ij}|e_{ij}\ra\la e_{ij}|\colon i,j\}$  then for any $i$ the result $I = i$ corresponds to $J \in A$ for some set $A$ of outcomes of $J$. Thus $\pp(X = x|I=i) = \pp(X=x|J \in A) = \sum_{j \in A} \pp(X=x|J=j)\pp(J = j|J \in A)$ and the concavity of the square root implies that to prove \eqref{eq:locking-fidelity}, it is enough to verify it with $J$ instead of $I$. In what follows we will therefore assume that each $M_i = \xi_i |e_i\ra\la e_i|$ for some $|e_i\ra \in \Ha$ and $\xi_i \in (0,1]$.

For $x \in [2^n]$ let $p(x) = \pp(X = x)$. Recall that $\max_{x}p(x) = 2^{-l}$.
By Born's rule we have
\begin{align*}
\pp(I = i|X =x) &= \frac{\xi_i}{td_B}\sum_{k=1}^{t}\sum_{b=1}^{d_B} \tr(U_k^\dag(|x\ra^A\la x|\otimes |b\ra^B\la b|)U_k |e_i\ra\la e_i|)\\
&=  \frac{\xi_i}{td_B} \sum_{k=1}^{t} p^A_{U_k|e_i\ra}(x).
\end{align*}
By the Bayes rule, we thus get
\begin{align}\label{eq:conditional_1}
\pp(X=x|I=i) = \frac{1}{\alpha} \frac{1}{t}\sum_{k=1}^{t} p^A_{U_k|e_i\ra}(x)p(x),
\end{align}
where
\begin{align}\label{eq:conditional_2}
\alpha = \frac{1}{t}\sum_{x'=1}^{2^n}\sum_{k=1}^{t} p^A_{U_k|e_i\ra}(x')p(x').
\end{align}

Now we get
\begin{align}\label{eq:intermediate-locking}
\sqrt{2}&D_H(\pp(X\in\cdot|I=i),\pp(X\in \cdot)) = \sqrt{\sum_{x=1}^{2^n} \Big(\sqrt{\frac{1}{t\alpha}\sum_{k=1}^{t} p^A_{U_k|e_i\ra}(x)p(x)} - \sqrt{p(x)}\Big)^2}\nonumber\\
&\le \sqrt{\sum_{x=1}^{2^n} \Big(\sqrt{\frac{1}{t\alpha}\sum_{k=1}^{t} p^A_{U_k|e_i\ra}(x)p(x)} - \sqrt{\frac{p(x)}{\alpha2^n}}\Big)^2} + \sqrt{\sum_{x=1}^{2^n} \Big(\sqrt{\frac{p(x)}{\alpha2^n}}- \sqrt{p(x)}\Big)^2}\nonumber\\
& \le \frac{2^{-l/2}}{\sqrt{\alpha}}\sqrt{\sum_{x= 1}^{2^n} \Big(\sqrt{\frac{1}{t}\sum_{k=1}^{t} p^A_{U_k|e_i\ra}(x)} - \frac{1}{2^{n/2}}\Big)^2}
+\frac{1}{\sqrt{\alpha}}\Big|\frac{1}{2^{n/2}} - \sqrt{\alpha}\Big|,
\end{align}
where in the first inequality we used the triangle inequality in $\ell_2$ and in the second one the upper bound on $p(x)$ coming from the min-entropy assumption.
Writing
\begin{displaymath}
\frac{1}{2^{n/2}} = \sqrt{\sum_{k=1}^{t} \frac{1}{t2^{n}}},
\end{displaymath}
and using the triangle inequality in $\ell_2$ (or just using the convexity of the function $s \mapsto (\sqrt{s}-2^{-n/2})^2$), we see that the first term on the right hand side above
is bounded from above by
\begin{align}\label{eq:intermediate-locking-2}
\frac{2^{-l/2}}{\sqrt{\alpha}} \sqrt{\frac{1}{t}\sum_{k=1}^{t}\sum_{x= 1}^{2^n} \Big(\sqrt{p^A_{U_k|e_i\ra}(x)} - \frac{1}{2^{n/2}}\Big)^2} \le \frac{2^{-l/2}}{\sqrt{\alpha}}\sqrt{2}\varepsilon,
\end{align}
where the last inequality follows from \eqref{eq:fur-Hellinger}.
Similarly,
\begin{align*}
\Big|\frac{1}{2^{n/2}} - \sqrt{\alpha}\Big| &\le \sqrt{\frac{1}{t}\sum_{k=1}^{t}\sum_{x'=1}^{2^n}p(x')\Big(\sqrt{p^A_{U_k|e_i\ra}(x')} - \frac{1}{2^{n/2}}\Big)^2}\\
&\le 2^{-l/2}\sqrt{2}\varepsilon.
\end{align*}

Combining the above estimates, we get
\begin{displaymath}
D_H(\pp(X\in\cdot|I=i),\pp(X\in \cdot)) \le 2\frac{2^{-l/2}\varepsilon}{\sqrt{\alpha}}.
\end{displaymath}
Since $\sqrt{\alpha} \ge \frac{1}{2^{n/2}}- \frac{\sqrt{2}\varepsilon}{2^{l/2}}$,  we get that if $2^{l-n} > 2\varepsilon^2$, then
\begin{displaymath}
D_H(\pp(X\in\cdot|I=i),\pp(X\in \cdot)) \le \frac{2\varepsilon}{2^{(l-n)/2} - \sqrt{2}\varepsilon},
\end{displaymath}
which proves the second part of the theorem. To prove the first part, note that when $X$ is uniformly distributed on $[t]$, we have $l = n$ and $\alpha = 2^{-n}$, so the second term on the righ hand side of \eqref{eq:intermediate-locking} vanishes, while \eqref{eq:intermediate-locking-2} shows that the first term is bounded by $\sqrt{2}\varepsilon$.
\end{proof}

\subsection{Further applications: data hiding \label{sec:further-applications}}

In this section we will present an example of further applications of Theorem \ref{th:main}, or more specifically Theorem \ref{thm:arbitrary-set}, which may be used to treat arbitrary subsets of the set of pure states. We will do it by discussing a particular application related to data hiding. We remark that this problem has been thoroughly studied in the literature with different types of theoretical guarantees (see e.g. \cite{985948,DiVincenzo2003,MR2094521}). Our goal is not to provide an extensive discussion, but rather to point out potential applications of uncertainty principles in this context.

Let us recall that in the problem of data hiding one aims at encoding a string of bits into a bipartite quantum state, shared between Alice and Bob in such a way that they cannot decode it by means of LOCC (local operations and classical communication) measurements (see \cite{MR2230995,Chitambar2014}), but are able to perfectly decode it provided that they can perform a global measurement.

In what follows just as in the locking protocol discussed before we will assume that the data $X$ to be hidden consists of $n$ random bits and that the distribution of $X$ is either uniform on $[2^n]$ or has sufficiently large min-entropy. We will show that it is possible to hide the data by using the locking scheme \eqref{eq:locking-scheme} of Fawzi-Hayden-Sen with just one random unitary matrix. More precisely, we have the following result.

Recall that a POVM $\{M_i\}$ on $\Ha = \Ha_A\otimes \Ha_B$ is called separable if each $M_i$ is of the form $M_i = \sum_j A_j\otimes B_j$ where $A_j,B_j$ are positive operators on $\Ha_A$ and $\Ha_B$ respectively. It is known that the class of separable measurements is strictly larger than the class of LOCC measurements.
\begin{theorem}\label{thm:data-hiding}
Assume that $n\ge 2\log_2(1/\varepsilon) + \Ci$, $d_A = 2^n$, $d_B\ge \Ci/\varepsilon^2$, where $\Ci$ is a sufficiently large universal constant. Let $U$ be a random unitary transformation of $\Ha$. Define
$\mathcal{E} \colon [2^n] \to \States(\Ha)$ with the formula
\begin{displaymath}
\mathcal{E}(x) =  \frac{1}{d_B}\sum_{b=1}^{d_B}U^\dag(|x\ra^A\la x|\otimes |b\ra^B\la b|)U.
\end{displaymath}
Then with probability at least $1 - 2\exp(\varepsilon^2 d/\Ci)$ the random encoding $\mathcal{E}$ has the following property.

Let $X$ be a random $n$-bit message. If $X$ is uniformly distributed on $[2^n]$, then for every separable measurement $\{M_i\}$ on $\mathcal{H} = \Ha_A\otimes \Ha_B$, if $I$ is the random variable describing the outcome of the measurement, we have
\begin{displaymath}
D_H(\pp(X \in \cdot|I = i), \pp(X\in \cdot)) \le \varepsilon.
\end{displaymath}
Moreover, if instead of uniform distribution of $X$, we assume that $H_{\min}(X) = l$ and $ 2^{l-n} > 2\varepsilon^2$, then
\begin{displaymath}
D_H(\pp(X \in \cdot|I = i), \pp(X\in \cdot)) \le \frac{2\varepsilon}{2^{(l-n)/2} -\sqrt{2}\varepsilon}.
\end{displaymath}
\end{theorem}

We can interpret the above theorem as locking against separable measurements or as a kind of data hiding, since clearly if $U$ is known, the message $X$ can be recovered by a global measurement. We remark that the security criterion provided in the theorem differs from those usually considered in data hiding literature and does not seem to be directly comparable to them. The point we would like to make here is the connection of data hiding with Bayesian type security criteria as above with uncertainty principles over the set of separable states. Note that the number of qubits one uses to encode $n$ bits is $n+ 2\log_2(1/\varepsilon) + O(1)$ and one uses just a single matrix $U$, instead of $t \simeq 1/\varepsilon^2$ matrices used in locking against all possible measurements in Theorem \ref{th:locking}. Let us note that random unitary matrices have been used in data hiding protocols e.g. in \cite{MR2094521}, however the setting therein was rather different from ours, first the Authors of \cite{MR2094521} were interested in the more difficult task of hiding qubits rather then bits and as a result their security criterion was of a different nature than ours; second, they were interested in approximate decoding; finally the number of qubits required to hide $n$ qubit states was of the order $2n +2\log\log n + 4\log(1/\varepsilon) + O(1)$, while the number of random unitaries used was of the order at least $Cn2^n/\varepsilon^2$.

The proof of Theorem \ref{thm:data-hiding} is analogous to the proof of Theorem \ref{th:locking}, the only difference is that in rank one decomposition of the measurement $\{M_i\}$ we can now assume that the states $|e_{ij}\ra$ and consequently $|e_i\ra$ are separable, which allows us to use estimates of Example \ref{ex:separable} with $t=1$. For this reason we will skip the details.

To finish this section, let us remark that as stated above, the splitting of the state $\mathcal{E}(x)$ between Alice and Bob is asymmetric since the dimensions $d_A$ and $d_B$ are different. However, by using a very similar approach one can consider a different bipartite structure on $\Ha$, splitting the system into parts of the same size $\sqrt{d}$ (the order of magnitude of $\E\sup_{|\psi\ra \in \Lambda}  G_{|\psi\ra}$ in Example  \ref{ex:separable} will be then of the order $d^{1/4} \le \sqrt{d_A}+\sqrt{d_B}$, so again one can take $t=1$ and obtain a meaningful bound).

\subsection{Euclidean embeddings \label{sec:embeddings}}

Our final application concerns finding almost Euclidean subspaces of the spaces $\ell_1^n(\ell_2^m)$. Recall that $\ell_1^n(\ell_2^m)$ is the space of $m\times n$ matrices, equipped with the norm $\|A\|_{\ell_1^n(\ell_2^m)} = \sum_{j=1}^n \sqrt{\sum_{i=1}^m |a_{ij}|^2}$ for $A = (a_{ij})_{1\le i \le n, 1\le j\le m}$.

To explain the context, let us recall that the famous theorem by Dvoretzky (in a version due to Milman, with best known quantitative bounds due to Schechtman \cite{MR2199631})  asserts that any $N$-dimensional Banach space $K$ contains a subspace $L$, such that $k = {\dim} L \ge \Ck^{-1}\frac{\varepsilon}{(\log \varepsilon^{-1})^2} \log N$, which is $(1+\varepsilon)$-isometric to a Hilbert space, i.e. there exists a linear map $T \colon \Ha \to L$, where $\Ha$ is a $k$ dimensional Hilbert space, such that for all $|x\ra \in \Ha$,
\begin{displaymath}
(1-\varepsilon)\||x\ra \| \le \| T |x\ra\| \le (1+\varepsilon)\||x\ra\|.
\end{displaymath}

The logarithmic dependence of $k$ on the dimension $N$ in general is best possible, however in certain cases it can be improved. In particular it is known that for spaces of the type $\ell_1^n(\ell_2^m)$ (in which case $N = nm$) one can find almost Euclidean subspaces of dimension $k \simeq \varepsilon^2 N$ (see e.g. \cite{Indyk2010}). It turns out that using Theorem \ref{th:main}, if $m$ is large one can improve the dependence on $\varepsilon$, as stated in Theorem \ref{thm:special-Dvoretzky} below. This improves the dependence on $\varepsilon$ with respect to presently known constructions. From the point of view of asymptotic convex geometry it would be of interest to verify if similar improvements can be obtained for more general Banach spaces.

We will now present in detail the relation between uncertainty principles and the construction of Euclidean subspaces of $\ell_1^n(\ell_2^m)$ as outlined in \cite{MR2932028}. Our improved uncertainty principles will translate into a better lower bound on the dimension of the almost Euclidean subspace. Consider $\Ha_A = \C^{d_A}$, $\Ha_B = \C^{d_B}$ and i.i.d. random unitary matrices $U_1,\ldots,U_t$. Introduce also an additional space $\Ha_C = \C^t$ and the space $K = \Ha_A\otimes \Ha_B \otimes \Ha_C$ of dimension $d_Ad_Bt$. Let also $|k\ra^C$, $k=1,\ldots,t$ be an orthonormal basis in $\Ha_C$. By identifying an element of $K$ with a linear map acting from $\Ha_A\otimes \Ha_C$ to $\Ha_B$ and considering its matrix representation in the bases $\{|a\ra^A\otimes|k\ra^C\colon a \in [d_A], k \in [t]\}$ and $\{|b\ra^B\colon b \in [d_B]\}$ of $\Ha_A\otimes \Ha_C$ and $\Ha_B$ respectively, we can identify $K$ with the space $\ell_1^n(\ell_2^m)$ with $n = d_At$, $m = d_B$. On $K$ we also have a natural Euclidean structure of a tensor product of Hilbert spaces.

Consider now a linear map $T\colon \Ha \to K$ defined as
\begin{displaymath}
T|\psi\ra = \frac{1}{\sqrt{t}}\sum_{k=1}^t (U_k|\psi\ra^{AB})|t\ra^C.
\end{displaymath}

Note that $T$ is an isometry between the Hilbert spaces $\Ha$ and $K$. At the same time, for any $|\psi\ra \in \Ha$,
\begin{align*}
\| T|\psi\ra\|_{\ell_1^n(\ell_2^m)} &= \frac{1}{\sqrt{t}}\sum_{k=1}^t\sum_{a=1}^{d_A} \sqrt{\sum_{b=1}^{d_B} |\la a|\la b|U_k|\psi\ra|^2} \\
&= \sqrt{d_At} \frac{1}{t}\sum_{k=1}^t F\Big(p^A_{U_k|\psi\ra},\unif([d_a])\Big)\\
& = \sqrt{d_At} (1 - Y_{|\psi\ra}^2),
\end{align*}
where we used the notation of Theorem \ref{th:main}. Equivalently, for all $|\psi\ra \in \Ha$,
\begin{displaymath}
\sqrt{1- \|\frac{1}{\sqrt{d_A t}} T|\psi\ra\|_{\ell_1^n(\ell_2^m)}} = Y_{|\psi\ra}.
\end{displaymath}

Denoting again $R = \E Y_{|\psi\ra}$ (recall that $R$ does not depend on $|\psi\ra$), by Theorem \ref{th:main} and Lemmas \ref{le:concentration}, \ref{le:Lipschitz-first} (as in the proof of Theorem \ref{th:fur}) we obtain that if $t = \lceil \Cp/\varepsilon^2\rceil$, then with high probability
\begin{displaymath}
R - \varepsilon \le \sqrt{1- \|\frac{1}{\sqrt{d_A t}} T|\psi\ra\|_{\ell_1^n(\ell_2^m)}} \le R + \varepsilon.
\end{displaymath}
Taking into account that the function $x\mapsto x^2$ is $2(R+\varepsilon)$ Lipschitz on $[0,R+\varepsilon]$, we obtain
\begin{displaymath}
1 - R^2 - 2(R+\varepsilon)\varepsilon \le \|\frac{1}{\sqrt{d_A t}} T|\psi\ra\|_{\ell_1^n(\ell_2^m)} \le 1 - R^2 +2(R+\varepsilon)\varepsilon.
\end{displaymath}

Since for $m = d_B \ge 2$, $R \le 1/\sqrt{m} \le 1/\sqrt{2}$, for $\varepsilon\le \varepsilon_0$ (where $\varepsilon_0>0$ is a universal constant), the left-hand side above is strictly positive and
\begin{displaymath}
  \frac{1 - R^2 +2(R+\varepsilon)\varepsilon}{1 - R^2 -2(R+\varepsilon)\varepsilon} \le 1 + \Cm(R+\varepsilon)\varepsilon = 1 + \Cm\frac{\varepsilon}{\sqrt{m}} + \Cm\varepsilon^2,
\end{displaymath}
which shows that an appropriate normalization of $T$ (call it $\tilde{T}$) satisfies
\begin{displaymath}
\||x\ra \| \le \| \tilde{T} |x\ra\|_{\ell_1^n(\ell_2^m)} \le (1 + \Cm\frac{\varepsilon}{\sqrt{m}} + \Cm\varepsilon^2)\||x\ra\|
\end{displaymath}
for all $|x\ra\in \Ha$.

Note that $\dim \Ha = \dim K/t \ge \varepsilon^2 (\dim K )/2\Cp$. Thus by a change of variables $\varepsilon^2 \to \varepsilon$ and an adjustment of constants the above implies

\begin{theorem}\label{thm:special-Dvoretzky}
There exists a universal constant $\Cn$ such that the following holds. Let $n, m$ be positive integers and define $N = nm$. Then the space $\ell_1^n(\ell_2^m)$ contains a subspace $H$ of dimension at least $\Cn^{-1}N\min(\varepsilon,\varepsilon^2 m)$, such that for a certain number $M$ and all $|x\ra \in H$,
\begin{displaymath}
  M \||x\ra\|_{\ell_2^N} \le \||x\ra\|_{\ell_1^n(\ell_2^m)} \le M(1+\varepsilon)\||x\ra\|_{\ell_2^N}.
\end{displaymath}
\end{theorem}
To our knowledge, the best dependence of $\dim H$ on the parameters $n,m,\varepsilon$ available in the literature so far is $\dim H \simeq N\varepsilon^2$ \cite{Indyk2010} mentioned above. Thus the bound of Theorem \ref{thm:special-Dvoretzky} provides an improvement if $m$ is large. In particular for $m \ge 1/\varepsilon$ we get $\dim H \simeq N\varepsilon \gg N\varepsilon^2$.

\section{Concluding remarks}

In this work we study fidelity uncertainty relations, which strengthen the total variation uncertainty relations due to Fawzi-Hayden-Sen \cite{MR2932028}. We show that for random unitary matrices such uncertainty principles hold with high probability over the product of unitary groups provided that the number of unitaries and the dimension of the ancilla are high enough (Theorem \ref{th:fur}). The parameters we obtain are better than those known from \cite{MR2932028} in the case of total variation uncertainty relations. The main tool we use is the Majorizing Measure Theorem due to Talagrand, which allows us to study uniform uncertainty estimates over general sets of pure states and obtain bounds in terms of their Gaussian mean width (Theorem \ref{th:main}). We also show that our estimates are essentially optimal even for uncertainty principles expressed in the total variation distance (Theorem \ref{thm:lower_bound}).

Fidelity uncertainty relations are subsequently applied to derive locking bounds for the Fawzi-Hayden-Sen protocol expressed in terms of the Hellinger distance. Even though the estimates we have obtained use a stronger proximity measure than in the case of \cite{MR2932028}, they improve the dependence of the key length on the min entropy of the message and thus apply to a larger subset of messages. They can also be shown to be optimal up to universal additive constants (Propositions \ref{prop:locking-optimality-general}, \ref{prop:locking-optimality-random}).

The general case of our Theorem \ref{th:main} is also used to provide results concerning data hiding schemes using a single random unitary matrix and providing a Bayes type security guarantee, expressed in terms of the Hellinger distance of the a posteriori distribution of the message (obtained via separable measurements) to the a priori distribution.

Finally, we apply Theorem \ref{th:main} to obtain estimates on the geometric problem of embedding the Euclidean space into the matricial space $\ell_1^n(\ell_2^m)$, with parameter dependence improved with respect to known results (Theorem \ref{thm:special-Dvoretzky}). This provides another example of connections between quantum information theory, asymptotic convex geometry and high dimensional probability, which have proved extremely fruitful in the last decade.

\paragraph{Acknowledgements} The author would like to thank Rafa{\l} Lata{\l}a, Gideon Schechtman and Karol {\.Z}yczkowski for instructive conversations concerning the topics investigated in this article.

\bibliographystyle{amsplain}
\bibliography{uncertainty}

\providecommand{\bysame}{\leavevmode\hbox to3em{\hrulefill}\thinspace}
\providecommand{\MR}{\relax\ifhmode\unskip\space\fi MR }
\providecommand{\MRhref}[2]{%
  \href{http://www.ams.org/mathscinet-getitem?mr=#1}{#2}
}
\providecommand{\href}[2]{#2}
\begin{thebibliography}{10}

\bibitem{ALPZ}
R.~Adamczak, R.' Lata{\l}a, Z.~Pucha{\l}a, and K.~{\.Z}yczkowski,
  \emph{Asymptotic entropic uncertainty relations}, Journal of Mathematical
  Physics \textbf{57} (2016), no.~3.

\bibitem{MR2760897}
G.W. Anderson, A.~Guionnet, and O.~Zeitouni, \emph{An introduction to random
  matrices}, Cambridge Studies in Advanced Mathematics, vol. 118, Cambridge
  University Press, Cambridge, 2010.

\bibitem{MR2605015}
G.~Aubrun, S.~Szarek, and E.~Werner, \emph{Nonadditivity of {R}\'enyi entropy
  and {D}voretzky's theorem}, J. Math. Phys. \textbf{51} (2010), no.~2, 022102,
  7.

\bibitem{MR2802300}
\bysame, \emph{Hastings's additivity counterexample via {D}voretzky's theorem},
  Comm. Math. Phys. \textbf{305} (2011), no.~1, 85--97.

\bibitem{MR3139428}
G.~Aubrun, S.~J. Szarek, and D.~Ye, \emph{Entanglement thresholds for random
  induced states}, Comm. Pure Appl. Math. \textbf{67} (2014), no.~1, 129--171.

\bibitem{MR1491097}
K.~Ball, \emph{An elementary introduction to modern convex geometry}, Flavors
  of geometry, Math. Sci. Res. Inst. Publ., vol.~31, Cambridge Univ. Press,
  Cambridge, 1997, pp.~1--58.

\bibitem{MR2230995}
I.~Bengtsson and K.~{\.Z}yczkowski, \emph{Geometry of quantum states},
  Cambridge University Press, Cambridge, 2006.

\bibitem{MR2451289}
P.~Bourgade, C.~P. Hughes, A.~Nikeghbali, and M.~Yor, \emph{The characteristic
  polynomial of a random unitary matrix: a probabilistic approach}, Duke Math.
  J. \textbf{145} (2008), no.~1, 45--69.

\bibitem{MR3185453}
S.~Brazitikos, A.~Giannopoulos, P.~Valettas, and B.~Vritsiou, \emph{Geometry of
  isotropic convex bodies}, Mathematical Surveys and Monographs, vol. 196,
  American Mathematical Society, Providence, RI, 2014.

\bibitem{MR520217}
S.~Chevet, \emph{S\'eries de variables al\'eatoires gaussiennes \`a valeurs
  dans {$E\hat \otimes _{\varepsilon }F$}. {A}pplication aux produits d'espaces
  de {W}iener abstraits}, S\'eminaire sur la {G}\'eom\'etrie des {E}spaces de
  {B}anach (1977--1978), \'Ecole Polytech., Palaiseau, 1978, pp.~Exp. No. 19,
  15.

\bibitem{Chitambar2014}
E.~Chitambar, D.~Leung, L.~Man{\v{c}}inska, M.~Ozols, and A.~Winter,
  \emph{Everything you always wanted to know about locc (but were afraid to
  ask)}, Communications in Mathematical Physics \textbf{328} (2014), no.~1,
  303--326.

\bibitem{2015arXiv151104857C}
P.~J. {Coles}, M.~{Berta}, M.~{Tomamichel}, and S.~{Wehner}, \emph{{Entropic
  Uncertainty Relations and their Applications}}, ArXiv e-prints (2015).

\bibitem{DiVincenzo2003}
D.~P. DiVincenzo, P.~Hayden, and B.~M. Terhal, \emph{Hiding quantum data},
  Foundations of Physics \textbf{33} (2003), no.~11, 1629--1647.

\bibitem{PhysRevLett.92.067902}
D.~P. DiVincenzo, M.~Horodecki, D.~W. Leung, J.~A. Smolin, and B.~M. Terhal,
  \emph{Locking classical correlations in quantum states}, Phys. Rev. Lett.
  \textbf{92} (2004), 067902.

\bibitem{985948}
D.~P. DiVincenzo, D.~W. Leung, and B.~M. Terhal, \emph{Quantum data hiding},
  IEEE Transactions on Information Theory \textbf{48} (2002), no.~3, 580--598.

\bibitem{Locking-decoding}
F.~Dupuis, J.~Florjanczyk, P.~Hayden, and D.~Leung, \emph{The locking-decoding
  frontier for generic dynamics}, Proceedings. Mathematical, Physical, and
  Engineering Sciences / The Royal Society \textbf{469} (2013), no.~2159.

\bibitem{MR2932028}
O.~Fawzi, P.~Hayden, and P.~Sen, \emph{From low-distortion norm embeddings to
  explicit uncertainty relations and efficient information locking},
  S{TOC}'11---{P}roceedings of the 43rd {ACM} {S}ymposium on {T}heory of
  {C}omputing, ACM, New York, 2011, pp.~773--782. \MR{2932028}

\bibitem{MR2371614}
Y.~Gordon, A.~E. Litvak, S.~Mendelson, and A.~Pajor, \emph{Gaussian averages of
  interpolated bodies and applications to approximate reconstruction}, J.
  Approx. Theory \textbf{149} (2007), no.~1, 59--73.

\bibitem{Hastings}
M.~B. {H}astings, \emph{Superadditivity of communication capacity using
  entangled inputs}, Nat. Phys. 5, 255 (2009).

\bibitem{MR2094521}
P.~Hayden, D.~Leung, P.~W. Shor, and A.~Winter, \emph{Randomizing quantum
  states: constructions and applications}, Comm. Math. Phys. \textbf{250}
  (2004), no.~2, 371--391.

\bibitem{Indyk2010}
P.~Indyk and S.~Szarek, \emph{Almost-euclidean subspaces of $\ell_{1}^n$ via
  tensor products: A simple approach to randomness reduction}, pp.~632--641,
  Springer Berlin Heidelberg, 2010.

\bibitem{MR2149924}
B.~Klartag and S.~Mendelson, \emph{Empirical processes and random projections},
  J. Funct. Anal. \textbf{225} (2005), no.~1, 229--245.

\bibitem{MR1267715}
R.~Lata{\l}a and K.~Oleszkiewicz, \emph{On the best constant in the
  {K}hinchin-{K}ahane inequality}, Studia Math. \textbf{109} (1994), no.~1,
  101--104.

\bibitem{MR1849347}
M.~Ledoux, \emph{The concentration of measure phenomenon}, Mathematical Surveys
  and Monographs, vol.~89, American Mathematical Society, Providence, RI, 2001.

\bibitem{MR1102015}
M.~Ledoux and M.~Talagrand, \emph{Probability in {B}anach spaces}, Ergebnisse
  der Mathematik und ihrer Grenzgebiete (3), vol.~23, Springer-Verlag, Berlin,
  1991.

\bibitem{1742-6596-143-1-012008}
D.~Leung, \emph{A survey on locking of bipartite correlations}, Journal of
  Physics: Conference Series \textbf{143} (2009), no.~1, 012008.

\bibitem{2016arXiv160506556L}
D.~J. {Lum}, M.~S. {Allman}, T.~{Gerrits}, C.~{Lupo}, V.~B. {Verma},
  S.~{Lloyd}, S.~W. {Nam}, and J.~C. {Howell}, \emph{{A Quantum Enigma Machine:
  Experimentally Demonstrating Quantum Data Locking}}, ArXiv e-prints (2016).

\bibitem{MR932170}
H.~Maassen and J.~B.~M. Uffink, \emph{Generalized entropic uncertainty
  relations}, Phys. Rev. Lett. \textbf{60} (1988), no.~12, 1103--1106.

\bibitem{MR3109633}
Elizabeth~S. Meckes and Mark~W. Meckes, \emph{Spectral measures of powers of
  random matrices}, Electron. Commun. Probab. \textbf{18} (2013), no. 78, 13.

\bibitem{MR2373017}
S.~Mendelson, A.~Pajor, and N.~Tomczak-Jaegermann, \emph{Reconstruction and
  subgaussian operators in asymptotic geometric analysis}, Geom. Funct. Anal.
  \textbf{17} (2007), no.~4, 1248--1282.

\bibitem{MR856576}
V.~D. Milman and G.~Schechtman, \emph{Asymptotic theory of finite-dimensional
  normed spaces}, Lecture Notes in Mathematics, vol. 1200, Springer-Verlag,
  Berlin, 1986, With an appendix by M. Gromov.

\bibitem{MR1008729}
G.~Schechtman, \emph{A remark concerning the dependence on {$\epsilon$} in
  {D}voretzky's theorem}, Geometric aspects of functional analysis (1987--88),
  Lecture Notes in Math., vol. 1376, Springer, Berlin, 1989, pp.~274--277.

\bibitem{MR2199631}
\bysame, \emph{Two observations regarding embedding subsets of {E}uclidean
  spaces in normed spaces}, Adv. Math. \textbf{200} (2006), no.~1, 125--135.

\bibitem{MR906527}
M.~Talagrand, \emph{Regularity of {G}aussian processes}, Acta Math.
  \textbf{159} (1987), no.~1-2, 99--149.

\bibitem{MR3184689}
\bysame, \emph{Upper and lower bounds for stochastic processes}, Ergebnisse der
  Mathematik und ihrer Grenzgebiete. 3. Folge., vol.~60, Springer, Heidelberg,
  2014.

\bibitem{MR2602484}
S.~Wehner and A.~Winter, \emph{Entropic uncertainty relations---a survey}, New
  J. Phys. \textbf{12} (2010), no.~February, 025009, 22.

\end{thebibliography}
\end{document}